\documentclass[%
 reprint,
 amsmath,amssymb,
 aps,
]{revtex4-1}

\usepackage{graphicx}% Include figure files
\usepackage{dcolumn}% Align table columns on decimal point
\usepackage{bm,bbm}% bold math
\usepackage{blindtext}
\usepackage{amssymb}
\usepackage{amsfonts}
\usepackage{amsthm}
\usepackage{tikz}
\usepackage{caption}
\usepackage{graphics}
\usepackage{xcolor}
\usepackage{shadethm}
\usepackage{subfigure}
\usepackage{epigraph}
\usetikzlibrary{snakes}
\usepackage[new]{old-arrows}

\usepackage{enumerate}
\usepackage{verbatim}
\usetikzlibrary{trees}
\usetikzlibrary{decorations.markings}
\usetikzlibrary{arrows}
\usetikzlibrary{cd}
\usepackage{mathrsfs}
\usepackage{xparse}
\numberwithin{equation}{section}

\usetikzlibrary{positioning,arrows,patterns}
\usetikzlibrary{decorations.markings}
\usetikzlibrary{calc}
\tikzset{
  photon/.style={decorate, decoration={snake}, draw=black},
  fermion/.style={draw=black, postaction={decorate},decoration={markings,mark=at position .55 with {\arrow{>}}}},
  fermion2/.style={dashed, dash phase=0.1pt, draw=black, postaction={decorate},decoration={markings,mark=at position .55 with {\arrow{>}}}},
  vertex/.style={draw,shape=circle,fill=black,minimum size=5pt,inner sep=0pt},
particle/.style={thick,draw=black},
particle2/.style={thick,draw=blue},
avector/.style={thick,draw=black, postaction={decorate},
    decoration={markings,mark=at position 1 with {\arrow[black]{triangle 45}}}},
gluon/.style={decorate, draw=black,
    decoration={coil,aspect=0}}
 }
\NewDocumentCommand\semiloop{O{black}mmmO{}O{above}}
{%
\draw[#1] let \p1 = ($(#3)-(#2)$) in (#3) arc (#4:({#4+180}):({0.5*veclen(\x1,\y1)})node[midway, #6] {#5};)
}
\usepackage{hyperref}
\hypersetup{colorlinks=true, citecolor=blue, linkcolor=blue, urlcolor=blue} 

\theoremstyle{plain}
\newtheorem{thm}{Theorem}[section]

\theoremstyle{definition}

\newtheorem*{thm*}{Theorem}
\newtheorem*{lem*}{Lemma}
\newtheorem*{prop*}{Proposition}
\newtheorem*{cor*}{Corollary}
\newtheorem*{exe*}{Exercise}
\newtheorem*{defn*}{Definition}
\newtheorem{rem}[thm]{Remark}

\theoremstyle{remark}

\newcommand{\R}{\mathbb{R}}

\newcommand{\E}{\mathbb{E}}

\newcommand{\calY}{\mathcal{Y}}

\newcommand{\dd}{{\mathrm{d}}}

\newcommand{\id}{\mathrm{id}}

\DeclareMathOperator{\tr}{Tr}

\DeclareMathOperator{\gh}{gh}

\DeclareMathOperator{\End}{End}
\DeclareMathOperator{\Hom}{Hom}

\newcommand{\de}{\partial}

\newcommand{\calB}{\mathcal{B}}
\newcommand{\calH}{\mathcal{H}}
\newcommand{\calS}{\mathcal{S}}
\newcommand{\calC}{\mathcal{C}}

\newcommand{\calO}{\mathcal{O}}
\newcommand{\calL}{\mathcal{L}}
\newcommand{\calM}{\mathcal{M}}

\newcommand{\calF}{\mathcal{F}}

\def\gpd{\,\lower1pt\hbox{$\longrightarrow$}\hskip-.24in\raise2pt
               \hbox{$\longrightarrow$}\,}

\newcommand{\I}{\mathrm{i}}

\newcommand{\calV}{\mathcal{V}}

\newcommand{\Tr}{\textnormal{Tr}}

\begin{document}

%\preprint{APS/123-QED}

\title{On Quantum Obstruction Spaces and Higher Codimension Gauge Theories}% Force line breaks with \\

\author{Nima Moshayedi}
 \email{nima.moshayedi@math.uzh.ch}
\affiliation{Institut f\"ur Mathematik, Universit\"at Z\"urich, Winterthurerstrasse 190, CH-8057 Z\"urich
}%


%  \homepage{http://www.Second.institution.edu/~Charlie.Author}
% \affiliation{
%  Second institution and/or address\\
%  This line break forced% with \\
% }%
% \affiliation{
%  Third institution, the second for Charlie Author
% \affiliation{%
%  Authors' institution and/or address\\
%  This line break forced with \textbackslash\textbackslash
% }%

% \collaboration{CLEO Collaboration}%\noaffiliation

\date{\today}% It is always \today, today,

\begin{abstract}
Using the quantum construction of the BV-BFV method for perturbative gauge theories, we show that the obstruction for quantizing a codimension 1 theory is given by the second cohomology group with respect to the boundary BRST charge. Moreover, we give an idea for the algebraic construction of codimension $k$ quantizations in terms of $\E_k$-algebras and higher shifted Poisson structures by formulating a higher version of the quantum master equation.
\begin{description}
\item[Keywords] Quantum Field Theory, Gauge Theory, BV-BFV Formalism, Deformation Quantization,\\Extended Field Theory, Shifted Poisson Structures, Higher Categories

% You may use the \texttt{description} environment to structure your abstract;
% use the optional argument of the \verb+\item+ command to give the category of each item. 
\end{description}
\end{abstract}

\maketitle

\section{Introduction}
The Batalin--Vilkovisky formalism \cite{BV1,BV2,BV3} is a powerful method to deal with perturbative quantizations of local gauge theories. The extension of this formalism to manifolds with boundary combines the Lagrangian approach of the Batalin--Vilkovisky (BV) formalism in the bulk with the Hamiltonian approach of the Batalin--Fradkin--Vilkovisky (BFV) formalism \cite{FV1,BF1} on the boundary of the underlying source (spacetime) manifold. This construction is known as the BV-BFV formalism \cite{CMR1,CMR2,CattMosh1}. In particular, it describes a codimension 1 quantum gauge formalism. Within a classical gauge theory one is interested in describing the obstructions for it to be quantizable. The cohomological symplectic formulation suggests an operator quantization for the boundary action. To get a well-defined and consistent cohomology theory, one has to require that this induced operator squares to zero. This will lead to obstruction spaces for boundary theories by considering a deformation quantization of the boundary action in order to formulate a boundary version of the quantum master equation as the gauge-independence condition. 
We will show that the obstruction for the quantization of manifolds with boundary is controlled by the second cohomology group with respect to the cohomological vector field on the boundary fields.  
Moreover, we formulate a classical extension of higher codimension $k$ theories as in \cite{CMR1} which we call BF$^k$V theories. The coupling for each stratum, in fact, is easily extended in the classical setting (BV-BF$^k$V theories), whereas for the quantum setting it might be rather involved. In order to formulate a fully extended topological quantum field theory in the sense of Baez--Dolan \cite{BaezDolan1995} or Lurie \cite{Lurie2009}, the coupling is indeed necessary.
Since one layer of the quantum picture, namely the quantum master equation, is described in terms of deformation quantization, we can formulate an algebraic approach for the higher codimension extension in terms of $\E_k$- and $\mathbb{P}_k$-algebras \cite{Lurie2017,Safronov2018}. Here $\E_k$ denotes the $\infty$-operad of little $k$-dimensional disks \cite{Lurie2017,Kontsevich1999,FresseWillwacher2020}.
Moving to one codimension higher corresponds to the shift of the Poisson structure by $-1$ since the symplectic form is shifted by $+1$ (see \cite{PantevToenVaquieVezzosi2013} for the shifted symplectic setting). This is controlled by the operad $\mathbb{P}_k$ on codimension $k$ which corresponds to $(1-k)$-shifted Poisson structures \cite{CalaquePantevToenVaquieVezzosi2017,Safronov2017}.
Using this notion, we give some ideas for the quantization in higher codimension. Moreover, if one uses the notion of Beilinson--Drinfeld ($\mathbb{BD}$) algebras \cite{BeilinsonDrinfeld2004,CostelloGwilliamVol2}, in particular $\mathbb{BD}_0$- and $\mathbb{BD}_1$-algebras, one can try to consider the action of $\mathbb{P}_0\cong\mathbb{BD}_0/\hbar$ (for $\hbar\to 0$) on $\mathbb{P}_1\cong\mathbb{BD}_1/\hbar$ (for $\hbar\to0$) in order to capture the algebraic structure of the classical bulk-boundary coupling (see also \cite[Section 5]{Safronov2017}). Here $\cong$ denotes an isomorphism of operads. 
In general, one can define the $\mathbb{BD}_k$ operads to provide a certain interpolation between the $\mathbb{P}_k$ and $\E_k$ operads in the sense that they are graded Hopf \cite{LivernetPatras2008} differential graded (dg) operads over $\mathbf{K}[\![\hbar]\!]$, where $\hbar$ is of weight $+1$ and $\mathbf{K}$ a field of characteristic zero, together with the equivalences 
\[
\mathbb{BD}_k/\hbar\cong\mathbb{P}_k,\qquad \mathbb{BD}_k[\![\hbar^{-1}]\!]\cong \E_k(\!(\hbar)\!).
\]
The formality of the $\E_k$ operad \cite{Tamarkin2003,Kontsevich1999,FresseWillwacher2020} implies the equivalence $\mathbb{BD}_k\cong \mathbb{P}_k[\![\hbar]\!]$.
There is a formulation of a $\mathbb{BD}_2$-algebra in terms of brace algebras \cite{CalaqueWillwacher2015,Safronov2018} and one can show that there is in fact a quasi-isomorphism $\mathbb{P}_2\cong \mathbb{BD}_2/\hbar$ (for $\hbar\to 0$).
However, the notion of a $\mathbb{BD}_k$-algebra for $k\geq3$ in terms of braces is currently not defined, but there should not be any obstruction to do this. Using these operads, one can define a deformation quantization of a $\mathbb{P}_{k+1}$-algebra $A$ to be a $\mathbb{BD}_{k+1}$-algebra $A_\hbar$ together with an equivalence of $\mathbb{P}_{k+1}$-algebras $A_\hbar/\hbar\cong A$ (see \cite{CalaquePantevToenVaquieVezzosi2017,MelaniSafronov2018} for a detailed discussion). 
\subsection*{Notation and conventions}
We will denote functions on a manifold $M$ by $\calO(M)$. 
Vector fields on $M$ will be denoted by $\mathfrak{X}(M)$ and the space of differential $k$-forms on $M$ by $\Omega^k(M)$. We denote by $A[\![t]\!]$ the space of formal power series in a formal parameter $t$ with coefficients in some algebra $A$. The imaginary unit is denoted by $\I:=\sqrt{-1}$. If the manifolds are infinite-dimensional, they are usually \emph{Banach} or \emph{Fr\'echet} manifolds. The ring of integers will be denoted by $\mathbf{Z}$. Real and complex numbers will be denoted by $\mathbf{R}$ and $\mathbf{C}$ respectively. A general field of characteristic zero will be denoted by $\mathbf{K}$.
% and one should of course do a careful analysis as some notions are not easily translated from the finite-dimensional case.

\section{Obstruction spaces for quantization on manifolds with boundary}
\subsection{Classical BV theories}
\label{subsec:BV}
We start with the BV approach for the bulk theory. 
A BV manifold is a triple 
\[
(\calF,\calS,\omega)
\]
such that $\calF$ is a $\mathbf{Z}$-graded supermanifold, $\calS\in \calO(\calF)$ is an even function of degree 0, and $\omega\in \Omega^2(\calF)$ an odd symplectic form of degree $-1$. 
The $\mathbf{Z}$-grading corresponds to the \emph{ghost number} which we will denote by ``$\mathrm{gh}$''. The BV space of fields $\calF$ is usually given as the $(-1)$-shifted cotangent bundle of the BRST space of fields, i.e. $\calF_{\mathrm{BV}}:=T^*[-1]\calF_{\mathrm{BRST}}$. In many cases, $\calF$ is an infinite-dimensional Fr\'echet manifold. 
Denote by $Q$ the Hamiltonian vector field of $\calS$ of degree $+1$, i.e. $\iota_Q\omega=\delta\calS$, where $\delta$ denotes the de Rham differential on $\calF$.
If we denote by $(\enspace,\enspace)$ the odd Poisson bracket induced by the odd symplectic form $\omega$ (also called the \emph{anti bracket}, or \emph{BV bracket}), we get 
\[
Q=(\calS,\enspace).
\]
Note that, by definition, $Q$ is \emph{cohomological}, i.e. $Q^2=0$.
Moreover, $Q$ is a \emph{symplectic} vector field, i.e. $L_Q\omega=0$, where $L$ denotes the Lie derivative.
For a BV theory we require the \emph{classical master equation (CME)} 
\begin{equation}
Q(\calS)=(\calS,\calS)=0
\end{equation}
to hold.
The assignment $\Sigma\mapsto (\calF_\Sigma,\calS_\Sigma,\omega_\Sigma)$ of a (usually, closed compact oriented) manifold $\Sigma$ to a BV manifold is called a \emph{BV theory}. By the physical property of \emph{locality}, given a BV theory, we usually want to work over \emph{local functions} on $\calF_\Sigma$, which we denote by $\calO_{loc}(\calF_\Sigma)\subset\calO(\calF_\Sigma)$. These are defined by functions on $\calF_\Sigma$ of the form
\[
\Phi\mapsto \int_{x\in\Sigma}\mathscr{L}\big(x,\Phi(x),\de\Phi(x),\de^2\Phi(x),\ldots,\de^N\Phi(x)\big),
\]
where $\Phi\in\calF_{\Sigma}$ denotes some field configuration and $\mathscr{L}$ denotes the \emph{Lagrangian density} of the given theory which depends on $\Phi$ and higher derivatives for $N\in\mathbf{Z}_{>0}$.

\subsection{Examples of classical BV theories}
We want to give some examples of non-reduced classical BV theories. For the reduced case see \cite{CMR1}.

\subsubsection{Electrodynamics}
\label{subsec:Electrodynamics_BV}
We want to consider the (minimal) BV extension of classical Euclidean electrodynamics for a trivial $U(1)$-bundle. Let $\Sigma$ be a smooth oriented $n$-dimensional Riemannian manifold. Denote by $*\colon \Omega^j(\Sigma)\to \Omega^{n-j}(\Sigma)$ the \emph{Hodge star} induced by the metric on $\Sigma$. The BV space of fields is then given by the shifted cotangent bundle $T^*[-1]E_\Sigma$, where 
\[
E_\Sigma:=\Omega^1(\Sigma)\oplus\Omega^{n-2}(\Sigma)\oplus \Omega^0(\Sigma)[1].
\]
The first term of $E_\Sigma$ denotes the space of connections $A$ of a trivial $U(1)$-bundle over $\Sigma$, the second term denotes the the space of the Hamiltonian counterpart of those connections (i.e. the momentum), which we denote by $B$, and the third term denotes the space of ghost fields $c$. Hence, we have 
\begin{multline*}
\calF_\Sigma:=T^*[-1]E_\Sigma=\Omega^1(\Sigma)\oplus\Omega^{n-2}(\Sigma)\oplus\\ \oplus\Omega^0(\Sigma)[1]\oplus\Omega^{n-1}(\Sigma)[-1]\oplus\Omega^2(\Sigma)[-1]\oplus\Omega^n(\Sigma)[-2]
\end{multline*}
We denote a field in $\calF_\Sigma$ by $(A,B,c,A^+,B^+,c^+)$. Then the BV symplectic form is given by 
\[
\omega_\Sigma=\int_\Sigma\left(\delta A\land \delta A^++\delta B\land \delta B^++\delta c\land \delta c^+\right).
\]
The BV action is given by 
\[
\calS_\Sigma=\int_\Sigma\left(B\land F_A+\frac{1}{2}B\land *B+A^+\land \dd c\right),
\]
where $F_A:=\dd A$, denotes the curvature of the connection $A$. The cohomological vector field is given by 
\begin{multline*}
Q_\Sigma=\int_\Sigma\bigg(\dd c\land \frac{\delta}{\delta A}+\dd B\land \frac{\delta}{\delta A^+}+\\+(*B+\dd A)\land\frac{\delta}{\delta B^+}+\dd A^+\land \frac{\delta}{\delta c^+}\bigg).
\end{multline*}
In particular, we have the following symmetries:
\begin{align*}
    Q_\Sigma(A)&=\dd c,\\
    Q_\Sigma (A^+)&=\dd B,\\
    Q_\Sigma (B^+)&=*B+\dd A,\\
    Q_\Sigma (c^+)&=\dd A^+,
\end{align*}
and $Q_\Sigma(B)=Q_\Sigma(c)=0$. It is easy to show that $Q^2=0$ and that the CME is indeed satisfied, i.e. $Q_\Sigma(\calS_\Sigma)=0$.

\subsubsection{Yang--Mills theory}
\label{subsubsec:Yang-Mills_theory}
Consider an $n$-dimensional closed oriented compact smooth Riemannian manifold $\Sigma$. Let $\mathfrak{g}$ be the Lie algebra of a finite-dimensional simply connected Lie group $G$ endowed with a $\mathfrak{g}$-invariant inner product given by $\langle g,h\rangle:=\Tr(gh)$. Moreover, let $P$ be principal $G$-bundle over $\Sigma$ and assume for simplicity that $P$ is trivial. The BV space of fields is given by 
\begin{multline*}
    \calF_\Sigma:=\Omega^1(\Sigma)\otimes \mathfrak{g}\oplus\Omega^{n-2}(\Sigma)\otimes \mathfrak{g}\oplus\\\oplus \Omega^0(\Sigma)\mathfrak{g}[1]\oplus\Omega^{n-1}(\Sigma)\otimes \mathfrak{g}[-1]\oplus\\\oplus\Omega^2(\Sigma)\otimes \mathfrak{g}[-1]\oplus\Omega^n(\Sigma)\otimes \mathfrak{g}[-2]
\end{multline*}
We denote a field in $\calF_\Sigma$ by components $(A,B,c,A^+,B^+,c^+)$. The BV symplectic form is given by 
\[
\omega_\Sigma=\int_\Sigma\Tr\big(\delta A\land \delta A^++\delta B\land \delta B^++\delta c\land \delta c^+\big)
\]
and the BV action by 
\begin{multline*}
\calS_\Sigma=\int_\Sigma\tr\bigg(B\land F_A+\frac{1}{2}B\land *B+\\+A^+\land \dd_Ac+B^+\land [B,c]+\frac{1}{2}c^+\land[c,c]\bigg),
\end{multline*}
where $F_A:=\dd A+\frac{1}{2}[A,A]$ denotes the curvature of the connection $A$ and $\dd_A$ is the covariant derivative for $A$. The cohomological vector field is given by 
\begin{multline*}
    Q_\Sigma=\int_\Sigma\bigg(\dd_Ac\land\frac{\delta}{\delta A}+[B,c]\land \frac{\delta}{\delta B}+\frac{1}{2}[c,c]\land \frac{\delta}{\delta c}+\\+(\dd_AB+[A^+,c])\land \frac{\delta}{\delta A^+}+(F_A+*B+[B^+,c])\land \frac{\delta}{\delta B^+}+\\+(\dd_AA^++[B,B^+]+[c,c^+])\land \frac{\delta}{\delta c^+}\bigg).
\end{multline*}
In particular, we have the following symmetries:
\begin{align*}
    Q_\Sigma(A)&=\dd_Ac,\\
    Q_\Sigma(B)&=[B,c],\\
    Q_\Sigma(c)&=\frac{1}{2}[c,c],\\
    Q_\Sigma(A^+)&=\dd_AB+[A^+,c],\\
    Q_\Sigma(B^+)&=F_A+*B+[B^+,c],\\
    Q_\Sigma(c^+)&=\dd_AA^++[B,B^+]+[c,c^+].
\end{align*}

\subsubsection{Chern--Simons theory}
\label{subsubsec:Chern-Simons}
Let $\Sigma$ be a $3$-dimensional closed compact oriented smooth manifold and let $\mathfrak{g}$ be the Lie algebra of a Lie group $G$ endowed with an invariant inner product (e.g. a simple Lie algebra). Denote by $\Tr(gh)$ the Killing form for two elements $g,h\in \mathfrak{g}$.
The space of fields is given by graded connections on a principal $G$-bundle. For simplicity, we assume that the bundle is trivial. Then the BV space of fields is given by 
\[
\calF_\Sigma:=\Omega^\bullet(\Sigma)\otimes\mathfrak{g}[1]=\bigoplus_{j=0}^3\Omega^{j}(\Sigma)\otimes\mathfrak{g}[1].
\]
A field in $\calF_\Sigma$ will be denoted by the tuple $(c,A,A^+,c^+)$. Note that the ghost numbers are $1,0,-1,-2$ respectively. Consider the \emph{superfield} $\mathbf{A}:=c+A+A^++c^+$. Then the BV symplectic form is given by 
\[
\omega_\Sigma=\frac{1}{2}\int_\Sigma\Tr(\delta\mathbf{A}\land \delta\mathbf{A})=\int_\Sigma\Tr(\delta c\land \delta c^++\delta A\land \delta A^+). 
\]
The cohomological vector field is given by 
\begin{multline*}
Q_\Sigma=\int_\Sigma\Tr\left(\left(\dd\mathbf{A}+\frac{1}{2}[\mathbf{A},\mathbf{A}]\right)\land\frac{\delta}{\delta\mathbf{A}}\right)=\\
=\int_\Sigma\Tr\bigg(\dd_Ac\land\frac{\delta}{\delta A}+(F_A+[c,A^+])\land\frac{\delta}{\delta A^+}+\\+(\dd_A A^++[c,c^+])\land \frac{\delta}{\delta c^+}+\frac{1}{2}[c,c]\land \frac{\delta}{\delta c}\bigg).
\end{multline*}
In particular, we have the following symmetries:
\begin{align*}
    Q_\Sigma(A)&=\dd_Ac,\\
    Q_\Sigma(c)&=\frac{1}{2}[c,c],\\
    Q_\Sigma(A^+)&=F_A+[c,A^+],\\
    Q_\Sigma(c^+)&=\dd_AA^++[c,c^+].
\end{align*}
The BV action is given by 
\begin{multline*}
    \calS_\Sigma=\int_\Sigma\Tr\bigg(\frac{1}{2}\mathbf{A}\land\dd\mathbf{A}+\frac{1}{6}\mathbf{A}\land[\mathbf{A},\mathbf{A}]\bigg)=\\
    =\int_\Sigma\Tr\bigg(\frac{1}{2}A\land\dd A+\frac{1}{6}A\land[A,A]+\frac{1}{2}A^+\land\dd_Ac+\\+\frac{1}{2}c\land \dd_AA^++\frac{1}{2}c^+\land[c,c]\bigg)
\end{multline*}

\subsubsection{(Abelian) $BF$ theory}
\label{subsubsec:BF}
Let us first consider abelian $BF$ theory. Let $\Sigma$ be an $n$-dimensional closed compact oriented smooth manifold. The space of fields is given by 
\[
\calF_\Sigma:=\Omega^\bullet(\Sigma)[1]\oplus\Omega^\bullet(\Sigma)[n-2].
\]
We denote the superfields by $\mathbf{A}\in\Omega^\bullet(\Sigma)[1]$ and $\mathbf{B}\in\Omega^\bullet(\Sigma)[n-2]$. The BV symplectic form is then given by
\[
\omega_\Sigma=\int_\Sigma \delta\mathbf{A}\land \delta\mathbf{B}.
\]
The BV action is given by 
\[
\calS_\Sigma=\int_\Sigma\mathbf{B}\land \dd\mathbf{A}.
\]
The cohomological vector field is given by 
\[
Q_\Sigma=\int_\Sigma\left(\dd\mathbf{A}\land \frac{\delta}{\delta\mathbf{A}}+\dd\mathbf{B}\land \frac{\delta}{\delta\mathbf{B}}\right).
\]
Note that $Q_\Sigma(\mathbf{A})=\dd\mathbf{A}$ and $Q_\Sigma(\mathbf{B})=\dd\mathbf{B}$. Now let us consider the case of non-abelian $BF$ theory, i.e. we consider a finite-dimensional Lie algebra $\mathfrak{g}$ with invariant inner product. The BV space of fields is then given by 
\[
\calF_\Sigma:=\Omega^\bullet(\Sigma)\otimes \mathfrak{g}[1]\oplus\Omega^\bullet(\Sigma)\otimes\mathfrak{g}[n-2]\ni(\mathbf{A},\mathbf{B}).
\]
The BV symplectic form is given by 
\[
\omega_\Sigma=\int_\Sigma\Tr(\delta \mathbf{B}\land\delta\mathbf{A}).
\]
The cohomological vector field is given by 
\[
Q_\Sigma=\int_{\Sigma}\Tr\bigg(\bigg(\dd\mathbf{A}+\frac{1}{2}[\mathbf{A},\mathbf{A}]\bigg)\land \frac{\delta}{\delta\mathbf{A}}+\dd_\mathbf{A}\mathbf{B}\land\frac{\delta}{\delta\mathbf{B}}\bigg).
\]
In particular, we have the following symmetries:
\begin{align*}
    Q_\Sigma(\mathbf{A})&=\dd\mathbf{A}+\frac{1}{2}[\mathbf{A},\mathbf{A}],\\
    Q_\Sigma(\mathbf{B})&=\dd_\mathbf{A}\mathbf{B}.
\end{align*}
The BV action is given by 
\[
\calS_\Sigma=\int_\Sigma\Tr\bigg(\mathbf{B}\land\bigg(\dd\mathbf{A}+\frac{1}{2}[\mathbf{A},\mathbf{A}]\bigg)\bigg).
\]

\begin{rem}
\label{rem:AKSZ}
Note that non-abelian $BF$ theory reduces to the abelian one when $\mathfrak{g}=\mathbf{R}$. In fact, abelian $BF$ theory is given by two copies of abelian Chern--Simons theory (i.e. the theory described in \ref{subsubsec:Chern-Simons} when $\mathfrak{g}=\mathbf{R}$). Moreover, (abelian) $BF$ theory and Chern--Simons theory are examples of a more general type of theory, called \emph{AKSZ theory} \cite{AKSZ}, which forms a subclass for BV theories. Other examples of AKSZ theories include the \emph{Poisson sigma model} \cite{I,SS1,CF4,CF1}, \emph{Witten's $A$- and $B$-model} \cite{Witten1988a,AKSZ}, \emph{Rozansky--Witten theory} \cite{RozanskyWitten1997,QiuZabzine2009}, \emph{Donaldson--Witten theory} \cite{Witten1989,Ikeda2011}, the \emph{Courant sigma model} \cite{Roytenberg2005,CattaneoQiuZabzine2010}, and \emph{2D Yang--Mills theory} \cite{CMR2,IM}, 
\end{rem}

\subsection{Obstruction space in the bulk}
\label{subsec:bulk}
It is well known that the obstruction space for quantization in the BV formalism is given by the first cohomology group with respect to $Q$. See e.g. \cite{BarnichDelMonte2018} and references therein.

\begin{thm}
\label{thm:obstruction_BV}
The obstruction space for a BV theory to be quantizable is given by
\begin{equation}
\label{eq:complex1}
\mathrm{H}^1_{Q}(\calO_{loc}(\calF)).
\end{equation}
\end{thm}
\begin{proof}
Consider a deformation of the BV action $\calS$, denoted by $\calS_\hbar$, depending on $\hbar$ and consider its expansion as a formal power series
\begin{multline}
\calS_\hbar:=\calS_0+\hbar\calS_1+\hbar^2\calS_2+O(\hbar^3)\\
=\sum_{k\geq 0}\hbar^k\calS_k\in\calO_{loc}(\calF)[\![\hbar]\!],
\end{multline}
where each $\calS_k\in\calO_{loc}(\calF)$ for all $k\geq 0$ and $\lim_{\hbar\to 0}\calS_\hbar=\calS$, i.e. $\calS_0:=\calS$. Note that $\gh\calS_k=0$ for all $k\geq 0$ since $\gh \calS=0$. 
For the quantum BV picture (see e.g. \cite{S,R}) one should note that there is a canonical second order differential operator $\Delta$ on $\calO_{loc}(\calF)$ such that $\Delta^2=0$. It is called BV Laplacian (see \cite{Khudaverdian2004,Severa2006} for a mathematical exposure). In particular, if $\Phi^i$ and $\Phi^+_i$ denote field and anti-field respectively, one can define $\Delta$ as
\[
\Delta f=\sum_i(-1)^{\gh \Phi^i+1}f\left\langle\frac{\overleftarrow{\delta}}{\delta\Phi^i},\frac{\overleftarrow{\delta}}{\delta\Phi^+_i}\right\rangle,\quad f\in \calO_{loc}(\calF).
\]
We have denoted by $\frac{\overleftarrow{\delta}}{\delta\Phi^{i}}$ and $\frac{\overrightarrow{\delta}}{\delta\Phi^{i}}$ the left and right derivatives with respect to $\Phi^{i}$. An analogue version also holds for the anti-fields $\Phi^+_i$. In fact, we have 
\begin{align}
\frac{\overrightarrow{\delta}}{\delta\Phi^{i}} f&=(-1)^{\gh\Phi^{i}(\gh f+1)}f\frac{\overleftarrow{\delta}}{\delta\Phi^{i}},\\
\frac{\overrightarrow{\delta}}{\delta\Phi^+_{i}} f&=(-1)^{(\gh\Phi^{i}+1)(\gh f+1)}f\frac{\overleftarrow{\delta}}{\delta\Phi^+_{i}}.
\end{align}
\begin{rem}
To be precise, for our constructions we want to consider a (global) BV Laplacian on \emph{half-densities} on $\calF$. It can be shown that for any odd symplectic supermanifold $\calF$ there exists a supermanifold $\calM$ such that $\calF\cong T^*[1]\calM$. Then $\calO(\calF)\cong \calO(T^*[1]\calM)=\Gamma(\bigwedge^\bullet T\calM)$. The \emph{Berezinian bundle} on $\calF$ is given by 
\[
\mathrm{Ber}(\calF)\cong\bigwedge^\mathrm{top}T^*\calM\otimes\bigwedge^\mathrm{top}T^*\calM.
\]
The half-densities on $\calF$ are defined by
\[
\mathrm{Dens}^\frac{1}{2}(\calF):=\Gamma\left(\mathrm{Ber}(\calF)^{\frac{1}{2}}\right). 
\]
One can then show that there is a canonical operator $\Delta^\frac{1}{2}_\calF$ on $\mathrm{Dens}^\frac{1}{2}(\calF)$ such that it squares to zero \cite{Khudaverdian2004}. At this point one should mention that this is only canonical in the finite-dimensional setting. For the infinite-dimensional case this only holds after a suitable \emph{renormalization}.
We can define a Laplacian by 
\[
\Delta_\sigma f:=\frac{1}{\sigma}\Delta^\frac{1}{2}_\calF(f\sigma),\quad f\in\calO(\calF),
\]
where $\sigma$ is a non-vanishing reference half-density on $\calF$ which is $\Delta^\frac{1}{2}_\calF$-closed. Note that $(\Delta_\sigma)^2=0$. We usually just write $\Delta\equiv \Delta_\sigma$ without mentioning $\sigma$. 
\end{rem}

To observe gauge-independence in the BV formalism, one requires the \emph{quantum master equation (QME)}
\begin{equation}
\label{QME}
\Delta\exp\left(\calS_\hbar/\hbar\right)=0\Longleftrightarrow(\calS_\hbar,\calS_\hbar)+2\hbar\Delta \calS_\hbar=0
\end{equation}
to hold. Here we denote by $\Delta$ the BV Laplacian.
Solving \eqref{QME} for each order in $\hbar$, we get the system of equations
\begin{align}
\label{eq:CME}
(\calS_0,\calS_0)&=0,\\
\label{eq:first}
\Delta\calS_0&=(\calS_0,\calS_1),\\
\label{eq:second}
\Delta \calS_1&=(\calS_0,\calS_2)+\frac{1}{2}(\calS_1,\calS_1)\\
&\quad\vdots
\end{align}
Note that Equation \eqref{eq:CME} is the CME which we assume to hold. Then, using the CME and the formula
\[
\Delta(f,g)=(f,\Delta g)-(-1)^{\gh g}(\Delta f,g),\quad \forall f,g\in \calO_{loc}(\calF), 
\]
we get 
\[
0=\Delta(\calS_0,\calS_0)=(\calS_0,\Delta\calS_0)=Q(\Delta\calS_0).
\]
Hence $\Delta \calS_0$ is closed with respect to the coboundary operator $Q=(\calS_0,\enspace)$. Moreover, if we assume that it is also $Q$-exact, we get that there is some $\calS_1\in\calO_{loc}(\calF)$ such that $\Delta\calS_0=Q(\calS_1)=(\calS_0,\calS_1)$, which is exactly the statement of Equation \eqref{eq:first}. This will automatically imply that all the higher order equations hold. Indeed, if $\Delta\calS_1=(\calS_0,\calS_1)$ for some $\calS_1\in\calO_{loc}(\calF)$, we get
\begin{multline}
\label{eq:second_degree}
0=\Delta\underbrace{(\calS_0,\calS_1)}_{\Delta\calS_0}=(\Delta\calS_0,\calS_1)-(-1)^{\gh\calS_0}(\calS_0,\Delta\calS_1)\\
=((\calS_0,\calS_1),\calS_1)-(\calS_0,\Delta\calS_1),
\end{multline}
where we used $\Delta^2=0$. Using the graded Jacobi formula for the BV bracket, we get
\begin{multline}
((\calS_0,\calS_1),\calS_1)=(\calS_0,(\calS_1,\calS_1))-\\
-(-1)^{(\gh\calS_0-1)(\gh\calS_1-1)}(\calS_1,(\calS_0,\calS_1)).
\end{multline}
Furthermore, by graded commutativity of the BV bracket we have
\begin{multline}
((\calS_0,\calS_1),\calS_1)=\\
=-(-1)^{(\gh(\calS_0,\calS_1)-1)(\gh\calS_1-1)}(\calS_1,(\calS_0,\calS_1)).
\end{multline}
Now since 
\[
\gh(\calS_0,\calS_1)=\gh\calS_0+\gh\calS_1+\gh(\enspace,\enspace)
\]
we get 
\[
2((\calS_0,\calS_1),\calS_1)=(\calS_0,(\calS_1,\calS_1)).
\]
Hence, using Equation \eqref{eq:second_degree}, we get
\[
(\calS_0,\Delta\calS_1)=\left(\calS_0,\frac{1}{2}(\calS_1,\calS_1)\right).
\]
This will give us 
\[
\Delta\calS_1=\frac{1}{2}(\calS_1,\calS_1)+\text{$Q$-exact term},
\]
so we can find some $\calS_2\in \calO_{loc}(\calF)$ such that the $Q$-exact term is given by $Q(\calS_2)=(\calS_0,\calS_2)$. This implies that Equation \eqref{eq:second} holds. The higher order equations hold in a similar iterative computation. 
\end{proof}
%Thus we get 
%$$\Delta(\{\calS_0^\de,\calS_0^\de\})=\{\calS_0^\de,\Delta\calS_0^\de\}=\{\calS_0^\de,\{\calS_0^\de,\calS_1^\de\}\}=0.$$
%Here we have used that $\Delta(\{f,g\})=\{\Delta f,g\}-(-1)^{\gh(f)}\{f,\Delta g\}$ with $\calS_0^\de$ being of degree $+1$ and equation \eqref{eq:first}.
%Since $\Delta$ is a degree $+1$ operator and $\calS^\de_0$ was of degree $+1$, we get by equation \eqref{eq:first} that $\{\calS^\de_0,\calS^\de_1\}$ is of degree $+2$. Moreover, since $\Delta^2=0$ and using equation \eqref{eq:second}, we get 
%\begin{align*}
%0=\Delta (\underbrace{\{\calS_0^\de,\calS_1^\de\}}_{\Delta\calS_0^\de})&=\{\Delta\calS_0^\de,\calS_1^\de\}-(-1)^{\gh(\calS_0^\de)}\{\calS_0^\de,\Delta\calS_1^\de\}\\
%&=\{\{\calS_0^\de,\calS_1^\de\},\calS_1^\de\}-\left\{\calS_0^\de,\{\calS_0^\de,\calS_2^\de\}+\frac{1}{2}\{\calS_1^\de,\calS_1^\de\}\right\}\\
%&=\{\{\calS_0^\de,\calS_1^\de\},\calS_1^\de\}-\{\calS^\de_0,\{\calS_0^\de,\calS_2^\de\}\}-\frac{1}{2}\{\calS_0^\de,\{\calS_1^\de,\calS_1^\de\}\}.
%\end{align*}
%We have
%$$\{\{\calS_0^\de,\calS^\de_1\},\calS_1^\de\}=\{\calS_0^\de,\{\calS_1^\de,\calS_1^\de\}\}-(-1)^{(\gh(\calS_0^\de)+1)(\gh(\calS_1^\de)+1)}\{\calS_1^\de,\{\calS_0^\de,\calS_1^\de\}\},$$
%and 
%$$\{\calS_1^\de,\{\calS_0^\de,\calS_1^\de\}\}=-(-1)^{(\gh(\calS_0^\de)+1)(\gh(\{\calS_0^\de,\calS_1^\de\})+1)}\{\{\calS_0^\de,\calS_1^\de\},\calS_1^\de\}.$$
%Now since $\gh(\{f,g\})=\gh(f)+\gh(g)+\gh(\{-,-\})$, we get 
%$$\{\calS_0^\de,\{\calS_1^\de,\calS_1^\de\}\}=0$$

\subsection{Classical BV-BFV theories}
Let us describe the BFV approach for the space of boundary fields.
A BFV manifold is a triple 
\[
\left(\calF^\de,\omega^\de,Q^\de\right),
\]
where $\calF^\de$ is a $\mathbf{Z}$-graded supermanifold, $\omega^\de\in \Omega^2(\calF^\de)$ an even symplectic form of ghost number $0$ and $Q^\de$ cohomological and symplectic vector field of degree $+1$ with odd Hamiltonian function $\calS^\de\in\calO_{loc}(\calF^\de)$ of ghost number $+1$, i.e. $\iota_{Q^\de}\omega^\de=\delta\calS^\de$, where $\delta$ denotes the de Rham differential on $\calF
^\de$. Moreover, we want 
\[
Q^\de(\calS^\de)=\{\calS^\de,\calS^\de\}=0.
\]
We say that a BFV manifold is \emph{exact}, if there exists a primitive $1$-form $\alpha^\de$, such that $\omega^\de=\delta\alpha^\de$. A BV-BFV manifold over an exact BFV manifold $(\calF^\de,\omega^\de=\delta\alpha^\de,Q^\de)$ is a quintuple 
\[
(\calF,\omega,\calS,Q,\pi),
\]
where $\pi\colon\calF\to \calF^\de$ is a surjective submersion such that
\begin{itemize}
    \item $\delta\pi Q=Q^\de$,
    \item $\iota_Q\omega=\delta\calS+\pi^*\alpha^\de$.
\end{itemize}
A consequence of this definition is 
\begin{equation}
    \label{eq:mCME}
    Q(\calS)=\pi^*\left(2\calS^\de-\iota_{Q^\de}\alpha^\de\right)
\end{equation}
which is called the \emph{modified classical master equation (mCME)}. The assignment $\Sigma\mapsto (\calF_\Sigma,\calS_\Sigma,Q_\Sigma,\pi_\Sigma\colon \calF_\Sigma\to \calF^\de_{\de\Sigma})$ of a manifold $\Sigma$ with boundary $\de\Sigma$ to a BV-BFV manifold is called a \emph{BV-BFV theory}.

\subsection{Examples of classical BV-BFV theories}
\subsubsection{Electrodynamics}
Let everything be as in \ref{subsec:Electrodynamics_BV} with the difference that $\Sigma$ now has non-vanishing boundary $\de\Sigma$. The boundary BFV space of fields is then given by 
\[
\calF^\de_{\de\Sigma}:=\Omega^1(\de\Sigma)\oplus \Omega^{n-2}(\de\Sigma)\oplus\Omega^0(\de\Sigma)[1]\oplus\Omega^{n-1}(\de\Sigma)[-1].
\]
A field in $\calF_{\de\Sigma}$ will be denoted by $(A,B,c,A^+)$. If we denote by $i\colon \de\Sigma\hookrightarrow \Sigma$ the inclusion of the boundary, the surjective submersion $\pi\colon \calF_\Sigma\to \calF_{\de\Sigma}$ acts on the component fields as 
\begin{align*}
    \pi_\Sigma(A)&=i^*(A):=\mathbb{A},\\
    \pi_\Sigma(B)&=i^*(B):=\mathbb{B},\\
    \pi_\Sigma(c)&=i^*(c):=\mathbbm{c},\\
    \pi_\Sigma(A^+)&=i^*(A^+):=\mathbb{A}^+,
\end{align*}
and $\mathbb{B}^+:=\pi_\Sigma(B^+)=0=\pi_\Sigma(c^+)=:\mathbbm{c}^+$. The BFV symplectic form is given by
\[
\omega_{\de\Sigma}^\de=\delta\alpha^\de_{\de\Sigma}=\int_{\de\Sigma}\big(\delta \mathbb{B}\land \delta \mathbb{A}+\delta \mathbb{A}^+\land \delta \mathbbm{c}\big),
\]
where 
\[
\alpha^\de_{\de\Sigma}=\int_{\de\Sigma}\big(\mathbb{B}\land \delta \mathbb{A}+\mathbb{A}^+\land \delta \mathbbm{c}\big).
\]
The BFV charge $Q^\de_{\de\Sigma}=\delta\pi_\Sigma Q_\Sigma$ is given by
\[
Q^\de_{\de\Sigma}=\int_{\de\Sigma}\left( \dd \mathbb{B}\land \frac{\delta}{\delta \mathbb{A}^+}+\dd \mathbbm{c}\land \frac{\delta}{\delta \mathbb{A}}\right).
\]
It is easy to check that the boundary action is given by 
\[
\calS^\de_{\de\Sigma}=\int_{\de\Sigma}\mathbbm{c}\land \dd \mathbb{B}.
\]
Then the mCME is indeed satisfied, since we have 
\begin{multline*}
    \iota_{Q_\Sigma}\omega_\Sigma=\int_\Sigma\big(\dd c\land \delta A^++\delta A\land \dd B+\\+\delta B\land (*B+\dd A)+\dd A^+\land \delta c\big)
\end{multline*}
and 
\begin{multline*}
    \delta \calS_\Sigma=\int_\Sigma\big(\delta B\land \dd A+B\land \dd\delta A+\delta B\land *B+\\+\delta A^+\land \dd c+A^+\land \dd\delta c\big).
\end{multline*}
Putting everything together and using Stokes' theorem, we get the claim.

\subsubsection{Yang--Mills theory}
Let everything be as in \ref{subsubsec:Yang-Mills_theory} with the difference that $\Sigma$ has non-vanishing boundary $\de\Sigma$. Let us denote the pullback of the forms $A,B,A^+,c$ with repsect to the inclusion $i\colon\de\Sigma\hookrightarrow \Sigma$ by $\mathbb{A},\mathbb{B},\mathbb{A}^+,\mathbbm{c}$ respectively. Note that here we have $\pi_\Sigma:=i^*\colon \calF_\Sigma\to\calF^\de_{\de\Sigma}$. The BFV space of fields is then given by 
\begin{multline*}
    \calF^\de_{\de\Sigma}=\Omega^1(\de\Sigma)\otimes \mathfrak{g}[1]\oplus\Omega^{n-2}(\de\Sigma)\otimes \mathfrak{g}[n-2]\oplus\\\oplus\Omega^0(\de\Sigma)\otimes\mathfrak{g}\oplus\Omega^{n-1}(\de\Sigma)\otimes \mathfrak{g}[n-2]. 
\end{multline*}
The BFV symplectic form is given by 
\[
\omega^\de_{\de\Sigma}=\delta\alpha^\de_{\de\Sigma}=\int_{\de\Sigma}\Tr\left(\delta \mathbb{B}\land \delta \mathbb{A}+\delta \mathbb{A}^+\land \delta \mathbbm{c}\right),
\]
where 
\[
\alpha^\de_{\de\Sigma}=\int_{\de\Sigma}\Tr\left(\mathbb{B}\land\delta \mathbb{A}+\mathbb{A}^+\land\delta \mathbbm{c}\right).
\]
The cohomological vector field is given by 
\begin{multline*}
Q^\de_{\de\Sigma}=\int_{\de\Sigma}\Tr\bigg(\dd_{\mathbb{A}}\mathbbm{c}\land\frac{\delta}{\delta \mathbb{A}}+[\mathbb{B},\mathbbm{c}]\land\frac{\delta}{\delta \mathbb{B}}+\\+(\dd_{\mathbb{A}}\mathbb{B}+[\mathbb{A}^+,\mathbbm{c}])\land \frac{\delta}{\delta \mathbb{A}^+}+\frac{1}{2}[\mathbbm{c},\mathbbm{c}]\land \frac{\delta}{\delta \mathbbm{c}}\bigg)
\end{multline*}
and the BFV action by 
\[
\calS^\de_{\de\Sigma}=\int_{\de\Sigma}\Tr\left(\mathbb{B}\land \dd_\mathbb{A}\mathbbm{c}+\frac{1}{2}\mathbb{A}^+\land[\mathbbm{c},\mathbbm{c}]\right).
\]

\subsubsection{Chern--Simons theory}
Let everything be as in \ref{subsubsec:Chern-Simons} with the difference that $\Sigma$ has non-vanishing boundary $\de\Sigma$. Let us denote the pullback of the forms $A,A^+,c$ with repsect to the inclusion $i\colon\de\Sigma\hookrightarrow \Sigma$ by $\mathbb{A},\mathbb{A}^+,\mathbbm{c}$ respectively. Note that here we have $\pi_\Sigma:=i^*\colon \calF_\Sigma\to\calF^\de_{\de\Sigma}$. We will denote the superfield on the boundary by $\mathfrak{A}:=\mathbbm{c}+\mathbb{A}+\mathbb{A}^+$. The BFV space of boundary fields is then given by 
\[
\calF^\de_{\de\Sigma}:=\Omega^\bullet(\de\Sigma)\otimes \mathfrak{g}[1]=\bigoplus_{j=0}^2\Omega^j(\de\Sigma)\otimes\mathfrak{g}[1]\ni\mathfrak{A}.
\]
The BFV symplectic form is then given by 
\begin{multline*}
\omega^\de_{\de\Sigma}=\delta\alpha^\de_{\de\Sigma}=\frac{1}{2}\int_{\de\Sigma}\Tr(\delta\mathfrak{A}\land \delta\mathfrak{A})=\\
=\int_{\de\Sigma}\Tr\bigg(\frac{1}{2}\delta \mathbb{A}\land \delta \mathbb{A}+\delta \mathbbm{c}\land \delta \mathbb{A}^+\bigg),
\end{multline*}
where 
\begin{multline*}
    \alpha^\de_{\de\Sigma}=\frac{1}{2}\int_{\de\Sigma}\Tr(\mathfrak{A}\land\delta\mathfrak{A})=\\
    =\frac{1}{2}\int_{\de\Sigma}\Tr\big(\mathbb{A}\land\delta \mathbb{A}+\mathbbm{c}\land \delta \mathbb{A}^++\mathbb{A}^+\land\delta \mathbbm{c}\big).
\end{multline*}
The cohomological vector field is given by 
\begin{multline*}
Q^\de_{\de\Sigma}=\int_{\de\Sigma}\Tr\bigg(\bigg(\frac{1}{2}\dd\mathfrak{A}+\frac{1}{2}[\mathfrak{A},\mathfrak{A}]\bigg)\land\frac{\delta}{\delta \mathfrak{A}}\bigg)=\\
=\int_{\de\Sigma}\Tr\bigg(\dd_\mathbb{A}\mathbbm{c}\land\frac{\delta}{\delta\mathbb{A}}+(F_\mathbb{A}+[\mathbbm{c},\mathbb{A}^+])\land\frac{\delta}{\delta \mathbb{A}^+}+\frac{1}{2}[\mathbbm{c},\mathbbm{c}]\land\frac{\delta}{\delta \mathbbm{c}}\bigg)
\end{multline*}
The BFV action is given by 
\begin{multline*}
    \calS^\de_{\de\Sigma}=\int_{\de\Sigma}\Tr\left(\frac{1}{2}\mathfrak{A}\land\dd\mathfrak{A}+\frac{1}{6}\mathfrak{A}\land[\mathfrak{A},\mathfrak{A}]\right)=\\
    =\int_{\de\Sigma}\Tr\left(\mathbbm{c}\land F_{\mathbb{A}}+\frac{1}{2}[\mathbbm{c},\mathbbm{c}]\land \mathbb{A}^+\right).
\end{multline*}

\subsubsection{(Abelian) $BF$ theory}
\label{subsec:BF(BFV)}
Let everything be as in \ref{subsubsec:BF} with the difference that $\Sigma$ has non-vanishing boundary $\de\Sigma$. Let $i\colon\de\Sigma\hookrightarrow \Sigma$ be the inclusion and denote the pullback of the superfields to the boundary by $\mathfrak{A}:=i^*(\mathbf{A})$ and $\mathfrak{B}:=i^*(\mathbf{B})$. Note that here we have $\pi_\Sigma:=i^*\colon \calF_\Sigma\to\calF^\de_{\de\Sigma}$.
Let us first look at the boundary theory for abelian $BF$ theory.
The BFV space of fields is given by 
\[
\calF^\de_{\de\Sigma}:=\Omega^\bullet(\de\Sigma)[1]\oplus\Omega^\bullet(\de\Sigma)[n-2]\ni(\mathfrak{A},\mathfrak{B}).
\]
The BFV symplectic form is given by 
\begin{equation*}
\label{eq:BFV_symplectic_abelian_BF}
\omega^\de_{\de\Sigma}=\delta\alpha^\de_{\de\Sigma}=\int_{\de\Sigma}\delta\mathfrak{A}\land\delta\mathfrak{B},
\end{equation*}
where 
\[
\alpha^\de_{\de\Sigma}=\int_{\de\Sigma}\mathfrak{A}\land \delta\mathfrak{B}.
\]
The cohomological vector field is given by 
\[
Q^\de_{\de\Sigma}=\int_{\de\Sigma}\bigg(\dd\mathfrak{A}\land\frac{\delta}{\delta\mathfrak{A}}+\dd\mathfrak{B}\land \frac{\delta}{\delta\mathfrak{B}}\bigg).
\]
The BFV action is given by 
\[
\calS^\de_{\de\Sigma}=\int_{\de\Sigma}\mathfrak{B}\land \dd\mathfrak{A}.
\]
For the non-abelian case we have the BFV space of fields
\[
\calF^\de_{\de\Sigma}:=\Omega^\bullet(\de\Sigma)\otimes\mathfrak{g}[1]\oplus\Omega^\bullet(\de\Sigma)\otimes \mathfrak{g}[n-2]\ni(\mathfrak{A},\mathfrak{B}).
\]
The BFV symplectic form is given by 
\[
\omega^\de_{\de\Sigma}=\delta\alpha^\de_{\de\Sigma}=\int_{\de\Sigma}\Tr\big(\delta\mathfrak{A}\land\delta\mathfrak{B}\big),
\]
where 
\[
\alpha^\de_{\de\Sigma}=\int_{\de\Sigma}\Tr\big(\mathfrak{A}\land \delta\mathfrak{B}\big).
\]
The cohomological vector field is given by 
\[
Q^\de_{\de\Sigma}=\int_{\de\Sigma}\Tr\bigg(\bigg(\dd\mathfrak{A}+\frac{1}{2}[\mathfrak{A},\mathfrak{A}]\bigg)\land\frac{\delta}{\delta\mathfrak{A}}+\dd_\mathfrak{A}\mathfrak{B}\land \frac{\delta}{\delta\mathfrak{B}}\bigg).
\]
The BFV action is given by 
\[
\calS^\de_{\de\Sigma}=\int_{\de\Sigma}\Tr\bigg(\mathfrak{B}\land\bigg(\dd\mathfrak{A}+\frac{1}{2}[\mathfrak{A},\mathfrak{A}]\bigg)\bigg).
\]

\subsection{Obstruction space on the boundary}
Similarly as for BV theories one can ask about the quantization obstruction for a BV-BFV theory, i.e. for a codimension 1 theory. In fact, we get the following theorem.
\begin{thm}
\label{thm:obstruction_BFV}
Let $(\calF,\omega,\calS,Q,\pi\colon \calF\to\calF^\de)$ be a BV-BFV manifold over an exact BFV manifold $(\calF^\de,\omega^\de=\delta\alpha^\de,Q^\de)$. The obstruction space for quantization on the underlying boundary BFV manifold $\calF^\de$ is given by
\begin{equation}
\label{eq:complex2}
\mathrm{H}^2_{Q^\de}(\calO_{loc}(\calF^\de)),
\end{equation}
where 
\[
Q^\de=\{\calS^\de,\enspace\}
\]
with $\{\enspace,\enspace\}$ the Poisson bracket induced by the symplectic form $\omega^\de$.
\end{thm}

\begin{proof}
Consider a deformation of the BFV action $\calS^\de$, denoted by $\calS^\de_\hbar$, depending on $\hbar$ and consider its expansion as a formal power series
\begin{multline}
\label{eq:expansion}
\calS^\de_\hbar:=\calS^\de_0+\hbar \calS^\de_1+\hbar^2\calS^\de_2+O(\hbar^3)\\=\sum_{k\geq 0}\hbar^k\calS^\de_k\in\calO_{loc}(\calF^\de)[\![\hbar]\!],
\end{multline}
where $\calS^\de_k\in\calO_{loc}(\calF^\de)$ for all $k\geq 0$ such that $\calS^\de_0:=\calS^\de$. Note that $\gh\calS^\de_k=+1$ since $\gh \calS^\de=+1$ and the corresponding symplectic form $\omega^\de$ is even of ghost number $0$. 
In the BV-BFV construction one assumes a symplectic splitting of the BV space of fields
\begin{equation}
\label{eq:splitting1}
\calF=\calB\times \calY
\end{equation}
where the BV symplectic form $\omega$ is constant on $\calB$. One should think of $\calB$ as the boundary part and $\calY$ as the bulk part of the fields. In fact, the space $\calB$ is constructed as the leaf space for a chosen polarization on the space of boundary fields $\calF^\de$ (i.e. a Lagrangian subbundle of $T\calF^\de$ closed under the Lie bracket) and $\calY$ is just a symplectic complement.
Using this splitting, we can write the mCME
    \begin{align}
    \label{eq:Y_part}
    \delta_\calY\calS&=\iota_{Q_\calY}\omega,\\
    \label{eq:B-part}
    \delta_\calB\calS&=-\alpha^\de,
    \end{align}
    where $Q_\calY$ denotes the part of the cohomological vector field $Q$ on $\calY$, $\delta_\calY$ and $\delta_\calB$ denote the corresponding parts of the de Rham differential $\delta$ on the BV space of fields $\calF$ according to the splitting \eqref{eq:splitting1}.
    Note that we have dropped the pullback $\pi^*$. These two equations together with \eqref{eq:mCME} imply
    \begin{equation}
    \label{eq:Y-part2}
    \frac{1}{2}(\calS,\calS)_\calY=\frac{1}{2}\iota_{Q_\calY}\iota_{Q_\calY}\omega=\calS^\de.
    \end{equation}
    Choose Darboux coordinates $(b^i,p_i)$ on $\calF^\de$ such that $b^i$ denotes the coordinates on the base $\calB$ and $p_i$ on the leaves.
    In the case of an infinite-dimensional Banach manifold, locally one has Darboux's theorem by using Moser's trick, whenever the tangent spaces are split, i.e. there are two Lagrangian subspaces $\calL_1,\calL_2\subset T_{\beta}\calF^\de$ such that $T_{\beta}\calF^\de=\calL_1\oplus \calL_2$ for some $\beta\in \calF^\de$. This is in general not true for Fr\'echet manifolds, even if the tangent spaces are split. Note that such a splitting is guaranteed if each fiber is a Hilbert space (see \cite{CattaneoContreras2018} for similar discussions).
    However, for the quantization we want to perturb around each critical point, thus we only have to use the linear structure. Additionally, we have to assume that the tangent spaces are split and that there is Darboux's theorem if we work in the Fr\'echet setting (see \cite{Kaveh2019} for discussions about Darboux's theorem on infinite-dimensional Fr\'echet manifolds).  
    This allows us to write 
    \[
    \alpha^\de=-\sum_ip_i\delta b^i.
    \]
    Using Equation \eqref{eq:B-part}, we get 
    \[
    \frac{\overrightarrow{\delta}}{\delta b^i}\calS=p_i,\quad \forall i.
    \]
    Denote by $\Delta_\calY$ the BV Laplacian restricted to $\calY$. We will assume that $\Delta_\calY\calS=0$. For the closed case this means that we assume that $\calS$ solves both, the CME and the QME. For the case with boundary, the BV Laplacian anyway only makes sense on $\calY$, so $\Delta=\Delta_\calY$. Next, we can obtain 
    \[
    \Delta_\calY\exp\left(\I\calS/\hbar\right)=\left(\frac{\I}{\hbar}\right)^2\frac{1}{2}(\calS,\calS)_\calY\exp\left(\I\calS/\hbar\right)
    \]
    and by Equation \eqref{eq:Y-part2}, we get 
    \begin{equation}
        \label{eq:mQME_Y-part}
        -\hbar^2\Delta_\calY\exp\left(\I\calS/\hbar\right)=\calS^\de\exp\left(\I\calS/\hbar\right).
    \end{equation}
    Now consider the standard quantization $\widehat{p}_i:=-\I\hbar\frac{\overrightarrow{\delta}}{\delta b^i}$. If $\widehat{p}_i$ acts on a function $\calS$ on $\calB$ parametrized by $\calY$, we get 
    \[
    \widehat{p}_i\calS=-\I\hbar p_i, \quad p_i\in \calY.
    \]
    Finally, considering the ordered standard quantization of $\calS^\de$ given by 
    \[
    \widehat{\calS^\de}:=\calS^\de\left(b^i,-\I\hbar\frac{\overrightarrow{\delta}}{\delta b^i}\right),
    \]
    where all the derivatives are placed to the right, and using Equation \eqref{eq:mQME_Y-part}, we get the \emph{modified quantum master equation (mQME)} \cite{CMR2} 
    \begin{equation}
    \label{eq:mQME1}
    \left(\hbar^2\Delta_\calY+\widehat{\calS^\de}\right)\exp\left(\I\calS/\hbar\right)=0.
    \end{equation}
In order to get a well-defined cohomology theory, we require that $$\left(\hbar^2\Delta+\widehat{\calS^\de}\right)^2=0.$$ 
Since $\Delta^2=0$ and obviously the commutator $\left[\Delta,\widehat{\calS^\de}\right]$ vanishes, we have to assume that $\left(\widehat{\calS^\de}\right)^2=0$. This clearly follows if
\begin{equation}
\label{eq:def_quant}
\calS^\de_\hbar\star\calS^\de_\hbar=0,
\end{equation}
where 
$$\star\colon \calO(\calF^\de)[\![\hbar]\!]\times \calO(\calF^\de)[\![\hbar]\!]\to \calO(\calF^\de)[\![\hbar]\!]$$
denotes the star product (deformation quantization) induced by the BFV form $\omega^\de$ and the standard ordering as mentioned above. Actually, the construction with the star product does not require the notion of a BV-BFV manifold and thus can be also considered independently for the BFV case. Moreover, the deformed boundary action $\calS^\de_\hbar$ satisfying \eqref{eq:def_quant} might spoil the mQME \eqref{eq:mQME1}.
Note that we can endow the deformed algebra $\calO_{loc}(\calF^\de)[\![\hbar]\!]$ with a dg structure by considering the differential given by 
\begin{equation}
    \label{eq:star_diff}
    Q^\de_\hbar:=\calS^\de_\hbar\star-.
\end{equation}
Then we have
\begin{multline}
Q^\de_\hbar(\calS^\de_\hbar)=\calS^\de_\hbar\star \calS^\de_\hbar=\calS^\de_\hbar\calS^\de_\hbar+\sum_{k\geq 1}\hbar^kB_k(\calS^\de_\hbar,\calS^\de_\hbar)\\
=\calS^\de_\hbar\calS^\de_\hbar+\hbar\{\calS^\de_\hbar,\calS^\de_\hbar\}+\hbar^2B_2(\calS^\de_\hbar,\calS^\de_\hbar)+O(\hbar^3),
\end{multline}
where $B_k$ denotes some bidifferential operator for all $k\geq 1$ with $B_1:=\{\enspace,\enspace\}$. 
Moreover, note that we have
\begin{multline}
\label{eq:Pbracket}
\{\calS^\de_\hbar,\calS^\de_\hbar\}=\{\calS^\de_0,\calS^\de_0\}+\hbar\{\calS^\de_0,\calS^\de_1\}+\hbar\{\calS^\de_1,\calS^\de_0\}+\\
+\hbar^2\{\calS^\de_1,\calS^\de_1\}+\hbar^2\{\calS^\de_0,\calS^\de_2\}+O(\hbar^3)
\end{multline}
Using \eqref{eq:expansion} and \eqref{eq:Pbracket}, we get
\begin{multline}
%\begin{split}
\calS_\hbar^\de\star \calS_\hbar^\de=(\calS^\de_0+\hbar\calS^\de_1+\hbar^2\calS^\de_2+O(\hbar^3))\times\\
\times(\calS^\de_0+\hbar\calS^\de_1+\hbar^2\calS^\de_2+O(\hbar^3))+\\
+\hbar(\{\calS^\de_0,\calS^\de_0\}+\hbar\{\calS^\de_0,\calS^\de_1\}+\hbar\{\calS^\de_1,\calS^\de_0\}+\\
+\hbar^2\{\calS^\de_1,\calS^\de_1\}+\hbar^2\{\calS^\de_0,\calS^\de_2\}+O(\hbar^3))+\hbar^2B_2(\calS^\de_0,\calS^\de_0)+O(\hbar^3)\\
=\calS^\de_0\calS^\de_0+\hbar(\calS^\de_1\calS^\de_0+\calS^\de_0\calS^\de_1+\{\calS^\de_0,\calS^\de_0\})+\\
+\hbar^2(\calS^\de_0\calS^\de_2+\calS^\de_2\calS^\de_0+\calS^\de_1\calS^\de_1+\{\calS^\de_0,\calS^\de_1\}+B_2(\calS^\de_0,\calS^\de_0))+O(\hbar^3)\\
\label{eq:defquant}
=\hbar\{\calS_0^\de,\calS_0^\de\}+\hbar^2(\{\calS_0^\de,\calS_1^\de\}+B_2(\calS^\de_0,\calS^\de_0))+O(\hbar^3),
%\end{split}
\end{multline}

were we have used the graded commutativity relation
$$\{f,g\}=-(-1)^{(\gh f+1)(\gh g+1)}\{g,f\},$$
the fact that $\{\enspace,\enspace\}$ is even of ghost number $0$ and that each $\calS^\de_k$ is odd of ghost number $+1$ for all $k\geq 0$. Note that by the CME for $\calS_0^\de$ the first term in \eqref{eq:defquant} vanishes. Moreover, by the associativity of the star product we get
$$\{\calS_0^\de,B_2(\calS^\de_0,\calS^\de_0)\}=Q^\de(B_2(\calS^\de_0,\calS^\de_0))=0,$$
and thus $B_2(\calS_0^\de,\calS_0^\de)$ is closed under the coboundary operator $Q^\de=\{\calS^\de,\enspace\}$. If we assume that $B_2(\calS_0^\de,\calS_0^\de)$ is also $Q^\de$-exact, there exists some $\calS_1^\de\in\calO_{loc}(\calF^\de)$, such that 
$$B_2(\calS_0^\de,\calS_0^\de)=-\{\calS_0^\de,\calS_1^\de\}=-Q^\de(\calS^\de_1).$$
Thus the coefficients in degree $+2$ vanish and one can check that by the construction of the star product all the higher coefficients will also vanish using a similar iterative procedure as we have seen before.
\end{proof}

More general, in the quantum BV-BFV construction \cite{CMR2} one can construct a geometric quantization \cite{Kir85,Wood97,BatesWeinstein2012} on the space of boundary fields $\calF^\de$ using the symplectic form $\omega^\de$ and the chosen polarization. This will give a vector space $\calH$ (actually a chain complex $\big(\calH,\widehat{\calS^\de}\big)$ associated to the source boundary. In fact, we can construct $\calH$ as the space of half-densities $\mathrm{Dens}^\frac{1}{2}(\calB)$ on $\calB$. We call $\widehat{\calH}:=\calH\hat{\otimes} \mathrm{Dens}^{\frac{1}{2}}(\calV)$ the \emph{space of states}. We have denoted by $\hat{\otimes}$ a certain completion of the tensor product in order to deal with the infinite-dimensional case.
Moreover, we have denoted by $\mathrm{Dens}^\frac{1}{2}(\calV)$ the space of half-densities on $\calV$. In order to deal with high energy terms for a functional integral quantization, we assume another splitting 
\begin{equation}
\label{eq:splitting2}
\calY=\calV\times\calY',
\end{equation}
where $\calV$ denotes the space of classical solutions (critical points) of the quadratic part of the action modulo gauge symmtery and $\calY'$ is a complement. In fact, we assume that the BV Laplacian and the BV symplectic form split accordingly as 

\begin{align}
    \Delta&=\Delta_\calV+\Delta_{\calY'},\\
    \omega&=\omega_\calV+\omega_{\calY'}.
\end{align}
Such a splitting is guaranteed for many important theories, such as (perturbations of) abelian $BF$ theories, by methods of \emph{Hodge decomposition} \cite{CMR2}.
In that special case, $\calB$ is in fact given by the fields restricted to the boundary.  The elements of $\calV$ are given by the \emph{zero modes} of the bulk fields and the elements of $\calY$ are given by the \emph{high energy} parts of the bulk fields. Choosing a gauge-fixing Lagrangian submanifold $\calL\subset \calY'$, a boundary state is given by 
\begin{equation}
    \widehat{\Psi}:=\int_{\calL\subset\calY'} \exp\left(\I\calS/\hbar\right)\in \widehat{\calH},
\end{equation}
where the functional integral is defined by its perturbative expansion.
One can then extend \eqref{eq:mQME1} to elements of $\widehat{\calH}$
\begin{equation}
    \label{eq:mQME2}
    \left(\hbar^2\Delta_\calV+\widehat{\calS^\de}\right)\widehat{\Psi}=0.
\end{equation}
Note that a state $\widehat{\Psi}$ depends on leaves in $\calB$ and zero modes in $\calV$. One can show that the space of zero modes is given by a finite-dimensional BV manifold $(\calV,\Delta_\calV,\omega_\calV)$ if we consider $BF$-like theories \cite{CMR2}. 
Note that in this case it makes sense to define $\Delta_\calV$.
As it was argued in \cite{CMR2}, there is a way of integrating out the zero modes. Using cutting and gluing techniques on the source, motivated by the constructions of \cite{Atiyah1988,Segal1988}, we will obtain a number which corresponds to the value of the partition function for a closed manifold. Moreover, in \cite{CMR2} it was shown that there is always a quantization $\widehat{\calS^\de}$ of $\calS^\de$ that squares to zero and satisfies \eqref{eq:mQME2}. It is fully described by integrals over the boundary of suitable configuration spaces determined by the underlying Feynman graphs.

\subsection{Examples of quantum BV-BFV theories}
One can extract the axiomatics for a quantum BV-BFV theory out of the computations we have seen before. A \emph{quantum BV-BFV theory} consists of the following data:
\begin{enumerate}[$(i)$]
    \item A graded vector space $\calH_{\widetilde{\Sigma}}$ associated to each $(n-1)$-dimensional manifold $\widetilde{\Sigma}$ with a choice of polarization on $\calF^\de_{\widetilde{\Sigma}}$. It is constructed by geometric quantization of the symplectic manifold $(\calF^\de_{\widetilde{\Sigma}},\omega^\de_{\widetilde{\Sigma}})$. The space $\calH_{\widetilde{\Sigma}}$ is called \emph{space of states}.
    \item A coboundary operator $\Omega_{\widetilde{\Sigma}}$ on $\calH_{\widetilde{\Sigma}}$ which is a quantization of the BFV action $\calS^\de_{\widetilde{\Sigma}}$. The operator $\Omega_{\widetilde{\Sigma}}$ is called \emph{quantum BFV operator}. 
    \item A finite-dimensional manifold $\calV_\Sigma$ associated to each $n$-dimensional manifold $\Sigma$, which is endwoed with a degree $-1$ symplectic form $\omega_{\calV_\Sigma}$ and a polarization on $\calF^\de_{\de\Sigma}$. It is called the \emph{space of residual fields}. Moreover, the space 
    \[
    \widehat{\calH}_{\Sigma}:=\calH_{\de\Sigma}\hat{\otimes}\mathrm{Dens}^\frac{1}{2}(\calV_\Sigma)
    \]
    is endowed with two commuting coboundary operators $\widehat{\Omega}_{\partial\Sigma}:=\Omega_{\de\Sigma}\otimes \id$ and $\widehat{\Delta}_\Sigma:=\id\otimes \Delta_{\calV_\Sigma}$, where $\Delta_{\calV_\Sigma}$ denotes the canonical BV Laplacian on half-densities on residual fields $\calV_\Sigma$.
    \item A \emph{state} $\widehat{\Psi}_\Sigma\in\widehat{\calH}_{\Sigma}$ which satisfies the \emph{modified QME}
    \[
    \big(\hbar^2\widehat{\Delta}_\Sigma+\widehat{\Omega}_{\de\Sigma}\big)\widehat{\Psi}_\Sigma=0.
    \]
\end{enumerate}

\subsubsection{Example: (Perturbations of) abelian $BF$ theory}
\label{subsubsec:quant_abelian_BF}
Consider the abelian version of the classical BV theory as in \ref{subsubsec:BF} and its classical BV-BFV extension. Moreover, we want to consider a source manifold $\Sigma$ whose boundary $\de\Sigma$ splits into the disjoint union of two boundary components $\de_1\Sigma$ and $\de_2\Sigma$ representing the \emph{incoming} and \emph{outgoing} boundary components respectively. On $\de_1\Sigma$ we choose the $\frac{\delta}{\delta\mathfrak{B}}$-polarization such that the quotient (leaf space) may be identified with $\calB_1:=\Omega^\bullet(\de_1\Sigma)[1]\ni\mathfrak{A}$ and on $\de_2\Sigma$ we choose the $\frac{\delta}{\delta\mathfrak{A}}$-polarization such that the quotient may be identified with $\calB_2:=\Omega^\bullet(\de_2\Sigma)[n-2]\ni\mathfrak{B}$. The whole leaf space is given by the product $\calB_{\de\Sigma}=\calB_{1}\times \calB_{2}$. The space of residual fields is given by the finite-dimensional BV manifold 
\[
\calV_\Sigma:=\mathrm{H}_{\mathrm{D}1}^\bullet(\Sigma)[1]\oplus\mathrm{H}_{\mathrm{D}2}^\bullet(\Sigma)[n-2],
\]
where $\mathrm{H}^\bullet_{\mathrm{D}j}(\Sigma)$ denotes the de Rham cohomology of the space $\Omega^\bullet_{\mathrm{D}j}(\Sigma):=\{\gamma\in\Omega^\bullet(\Sigma)\mid i_j^*\gamma=0\}$, where $i_j$ denotes the inclusion map $\de_j\Sigma\hookrightarrow \Sigma$. Here D stands for \emph{Dirichlet}. Note that by Poincar\'e duality we get $\calV_\Sigma=T^*[-1](\mathrm{H}^\bullet_{\mathrm{D}1}(\Sigma)[1])=T^*[-1](\mathrm{H}^\bullet_{\mathrm{D}2}(\Sigma)[n-2])$. The BV Laplacian $\Delta_{\calV_\Sigma}$ can be defined by using a basis $([\chi_i])_i$ of $\mathrm{H}^\bullet_{\mathrm{D}1}(\Sigma)$ and its dual basis $([\chi^i])_i$ of $\mathrm{H}^\bullet_{\mathrm{D}2}(\Sigma)$ with chosen representatives $\chi_i\in\Omega^\bullet_{\mathrm{D}1}(\Sigma) $ and $\chi^i\in\Omega^\bullet_{\mathrm{D}2}(\Sigma)$. Note that we have 
\[
\int_\Sigma\chi^i\land\chi_j=\delta^i_j,
\]
and we can write the residual fields in $\calV_\Sigma$ by
\[
\mathsf{a}=\sum_iz^i\chi_i,\quad \mathsf{b}=\sum_jz^+_j\chi^j,
\]
with $(z^i,z^+_j)$ being canonical coordinates on $\calV_\Sigma$. The BV symplectic form on $\calV_\Sigma$ is then given by 
\[
\omega_{\calV_\Sigma}=\sum_i(-1)^{1+(n-1)\gh z^i}\delta z_i^+\land \delta z^i.
\]
The BV Laplacian on $\calV_\Sigma$ is then given by 
\[
\Delta_{\calV_\Sigma}=\sum_i(-1)^{1+(n-1)\gh z^i}\frac{\de}{\de z^i}\frac{\de}{\de z^+_i}.
\]
The quantum BFV operator $\widehat{\Omega}_{\de\Sigma}$, acting on $\calB_{\de\Sigma}\times\calV_\Sigma$, is given by the ordered standard quantization of $\calS^\de_{\de\Sigma}$ relative to the chosen polarization:
\[
\widehat{\Omega}_{\de\Sigma}=\I\hbar(-1)^n\left(\int_{\de_2\Sigma}\dd\mathfrak{B}\land\frac{\delta}{\delta\mathfrak{B}}+\int_{\de_1\Sigma}\dd\mathfrak{A}\land\frac{\delta}{\delta\mathfrak{A}}\right).
\]
Using the \emph{effective action} given by 
\begin{multline*}
\calS^{\mathrm{eff}}_{\Sigma}:=(-1)^{n-1}\left(\int_{\de_2\Sigma}\mathfrak{B}\land\mathsf{a}-\int_{\de_1\Sigma}\mathsf{b}\land \mathfrak{A}\right)-\\-(-1)^{2n}\int_{\de_2\Sigma\times \de_1\Sigma}\pi_1^*\mathfrak{B}\eta\pi_2^*\mathfrak{A},
\end{multline*}
where $\eta\in \Omega^{n-1}(C_2(\Sigma))$ is a chosen \emph{propagator} on the compactified configuration space 
\[
C_2(\Sigma):=\overline{\{(x_1,x_2)\in \Sigma^2\mid x_1\not=x_2\}}^{\mathrm{FMAS}}
\]
(here FMAS stands for the \emph{Fulton--MacPherson}/\emph{Axelrod--Singer} compactification of configuration spaces \cite{FulMacPh,AS2}), we get the state
\[
\widehat{\Psi}_\Sigma=T_\Sigma\exp\left(\I\calS^\mathrm{eff}_\Sigma/\hbar\right).
\]
Here $T_\Sigma\in\mathbf{C}$ denotes a coefficient expressed in terms of the \emph{Reidemeister torsion}. Indeed, one can then immediatly check that the mQME is satisfied where $\widehat{\Delta}_\Sigma:=\Delta_{\calV_\Sigma}$ acts on the fibers of $\calB_{\de\Sigma}\times \calV_\Sigma$. 
We can construct the state space to be 
\[
\widehat{\calH}_{\Sigma}=\left(\prod_{j_1,j_2\geq 0}\calH^{j_2,n-2}_{\de_2\Sigma}\hat{\otimes}\calH^{j_1,1}_{\de_1\Sigma}\right)\hat{\otimes}\mathrm{Dens}^{\frac{1}{2}}(\calV_\Sigma), 
\]
where $\calH^{j,\ell}_{\de\Sigma}$ is the vector space of $j$-linear functionals on $\Omega^\bullet(\de\Sigma)[\ell]$ of the form 
\[
\Omega^\bullet(\de\Sigma)[\ell]\ni\mathfrak{D}\mapsto \int_{(\de\Sigma)^j}\gamma\pi_1^*\mathfrak{D}\land\dotsm \land \pi_j^*\mathfrak{D}
\]
times some prefactor (given in terms of the Reidemeister torsion), where $\gamma$ is some distributional form on $(\de\Sigma)^j$ and $\pi_i$ denotes the projection to the $i$-th component. Considering perturbations, we asymptotically get that states are of the form 
\begin{multline*}
    \widehat{\Psi}_\Sigma\sim T_\Sigma\exp\left(\I\calS^\mathrm{eff}_\Sigma/\hbar\right)\times\\
    \times\sum_{k\geq 0}\hbar^k\sum_{j_1,j_2\geq 0}\int_{(\de_1\Sigma)^{j_1}\times(\de_2\Sigma)^{j_2}}R^k_{j_1j_2}(\mathsf{a},\mathsf{b})\pi_{1,1}^*\mathfrak{A}\land\dotsm\\\dotsm\land\pi_{1,j_1}^*\mathfrak{A}\land\pi_{2,1}^*\mathfrak{B}\land\dotsm\land\pi_{2,j_2}^*\mathfrak{B},
\end{multline*}
where $\pi_{i,j}$ denotes the $j$-th projection of $(\de_i\Sigma)^{j_i}$ and $R^k_{j_1j_2}$ denotes distributional forms on $(\de_1\Sigma)^{j_1}\times(\de_2\Sigma)^{j_2}$ with values in $\mathrm{Dens}^\frac{1}{2}(\calV_\Sigma)$. Note that here $\calS^\mathrm{eff}_\Sigma$ is replaced by the corresponding zero-loop effective action. We refer the reader to \cite{CMR2} for more examples and a detailed discussion of perturbative quantizations on manifolds with boundary.

An important version of a perturbation of abelian $BF$ theory is given by \emph{split Chern--Simons theory} \cite{CMW}. Consider Chern--Simons theory as in \ref{subsubsec:Chern-Simons} for a Lie algebra $\mathfrak{g}$ endowed with an invariant pairing $\langle\enspace,\enspace\rangle$. Moreover, consider a suitable 3-manifold $\Sigma$.
As we have seen, for $\mathbf{A}\in\Omega^\bullet(\Sigma)\otimes\mathfrak{g}[1]$, the BV action is given by 
\[
\calS_\Sigma=\int_\Sigma \left(\frac{1}{2}\langle\mathbf{A},\dd\mathbf{A}\rangle+\frac{1}{6}\langle\mathbf{A},[\mathbf{A},\mathbf{A}]\rangle\right).
\]
Assume that the Lie algebra splits as $\mathfrak{g}=V\oplus W$ into maximally isotropic subspaces with respect to $\langle\enspace,\enspace\rangle$, i.e. the pairing restricts to zero on $V$ and $W$ and $\dim V=\dim W=\frac{1}{2}\dim\mathfrak{g}$. Then one can identify $W\cong V^*$ by using the pairing and consider a decomposition $\mathbf{A}=\mathbf{V}+\mathbf{W}$ with $\mathbf{V}\in\Omega^\bullet(\Sigma)\otimes V[1]$ and $\mathbf{W}\in\Omega^\bullet(\Sigma)\otimes W[1]$. Then we can decompose the action $\calS_\Sigma=\calS_\Sigma^\mathrm{kin}+\calS^\mathrm{int}_\Sigma$ into a \emph{kinetic} and \emph{interaction} part: 
\begin{align*}
\calS^{\mathrm{kin}}_\Sigma&=\frac{1}{2}\int_\Sigma\langle \mathbf{A},\dd\mathbf{A}\rangle=\int_\Sigma\langle \mathbf{W},\dd\mathbf{V}\rangle,\\
\calS^\mathrm{int}_\Sigma&=\frac{1}{6}\int_\Sigma\langle\mathbf{A},[\mathbf{A},\mathbf{A}]\rangle\\&=\frac{1}{6}\int_\Sigma\langle\mathbf{V}+\mathbf{W},[\mathbf{V}+\mathbf{W},\mathbf{V}+\mathbf{W}]\rangle.
\end{align*}
An important assumption on the theory is that $(\mathfrak{g},V,W)$ is in fact a \emph{Manin triple}, i.e. $V$ and $W$ are actually Lie subalgebras of $\mathfrak{g}$.
The quantum picture is then similar to the one of abelian $BF$ theory (see \cite{CMW} for a detailed construction). More general perturbations of abelian $BF$ theory is given by \emph{AKSZ theories} \cite{AKSZ} as mentioned in Remark \ref{rem:AKSZ} (see also \cite{CMR1,CMR2,CMW4} for a detailed treatment of such theories in the BV-BFV formalism).

\section{Higher codimension}

\subsection{Higher codimension gauge theories: BV-BF$^k$V theories}
Since the BV-BFV construction is a codimension 1 formulation, we have an action of a dg algebra of observables, coming from the deformation quantization construction, to a chain complex (or vector space) associated to the boundary via geometric quantization with respect to the symplectic manifold $(\calF^\de,\omega^\de)$. This corresponds to the action of the operator $\widehat{\calS^\de}\in \End(\calH)$ on $\widehat{\Psi}\in\widehat{\calH}$. 
Classical BV-BFV theories can be extended to higher codimension manifolds \cite{CMR1}. One can define an exact BF$^k$V manifold to be a triple $(\calF^{\de^k},\omega^{\de^k}=\delta\alpha^{\de^k},Q^{\de^k})$ where $\calF^{\de^k}$ is a $\mathbf{Z}$-graded supermanifold, $\omega^{\de^k}\in \Omega^2(\calF^{\de^k})$ is an exact symplectic form of ghost number $k-1$ with primitive 1-form $\alpha^{\de^k}$, and $Q^{\de^k}\in\mathfrak{X}(\calF^{\de^k})$ is a cohomological, symplectic vector field with Hamiltonian function $\calS^{\de^k}$ of ghost number $k$.
A BV-BF$^k$V manifold over an exact BF$^k$V manifold $(\calF^{\de^k},\omega^{\de^k}=\delta\alpha^{\de^k},Q^{\de^k})$ is a quintuple 
\[
(\calF^{\de^{k-1}},\omega^{\de^{k-1}},\calS^{\de^{k-1}},Q^{\de^{k-1}},\pi\colon \calF^{\de^{k-1}}\to \calF^{\de^{k}})
\]
such that $\pi$ is a surjective submersion and
\begin{itemize}
    \item $\delta\pi Q^{\de^{k-1}}=Q^{\de^k}$,
    \item $\iota_{Q^{\de^{k-1}}}\omega^{\de^{k-1}}=\delta\calS^{\de^{k-1}}+\pi^*\alpha^{\de^{k}}$.
\end{itemize}
Again, this will lead to a higher codimension version of the mCME
\begin{equation}
    \label{eq:higher_mCME}
    Q^{\de^{k-1}}\left(\calS^{\de^{k-1}}\right)=\pi^*\left(2\calS^{\de^k}-\iota_{Q^{\de^k}}\alpha^{\de^k}\right).
\end{equation}

\subsubsection{Example: classical codimension 2 theory}

Consider Yang--Mills theory as in \ref{subsubsec:Yang-Mills_theory}. Let $\Sigma_2\subset\Sigma$ be a codimension 2 stratum. The BV-BFV theory on $\Sigma$ and $\de\Sigma$ induces the following data associated on $\Sigma_2$: The space of fields
\[
\calF^{\de^2}_{\Sigma_2}=\Omega^{n-2}(\Sigma_2)\otimes\mathfrak{g}[n-2]\oplus\Omega^0(\Sigma_2)\otimes\mathfrak{g}[1]\ni(\mathbb{B}_2,\mathbbm{c}_2),
\]
with $\gh \mathbb{B}_2=0$ and $\gh \mathbbm{c}_2=1$. The $\mathrm{BF}^2\mathrm{V}$ symplectic form is given by  
\[
\omega^{\de^2}_{\Sigma_2}=\delta\alpha^{\de^2}_{\Sigma_2}=\int_{\Sigma_2}\tr\big(\delta\mathbb{B}_2\land\delta\mathbbm{c}_2\big),
\]
with 
\[
\alpha^{\de^2}_{\Sigma_2}=\int_{\Sigma_2}\tr\big(\mathbb{B}_2\land \delta\mathbbm{c}_2\big).
\]
The cohomological vector field is given by 
\[
Q^{\de^2}_{\Sigma_2}=\int_{\Sigma_2}\tr\bigg([\mathbb{B}_2,\mathbbm{c}_2]\land\frac{\delta}{\delta\mathbb{B}_2}+\frac{1}{2}[\mathbbm{c}_2,\mathbbm{c}_2]\land\frac{\delta}{\delta\mathbbm{c}_2}\bigg).
\]
The $\mathrm{BF}^2\mathrm{V}$ action (which one can obtain from $Q^{\de^2}_{\Sigma_2}$ by using the construction of \cite{Roytenberg2005}) is given by 
\[
\calS^{\de^2}_{\Sigma_2}=\int_{\Sigma_2}\tr\bigg(\frac{1}{2}\mathbb{B}_2\land[\mathbbm{c}_2,\mathbbm{c}_2]\bigg).
\]

\begin{rem}
The quantum extension is more difficult and requires certain algebraic constructions. 
Following the codimension 1 construction, one can try to formulate a similar procedure by considering a deformation quantization of Poisson structures with higher shifts. However, we will not consider a general coupling of higher codimension theories here, but rather describe the idea for quantization of the according BF$^k$V theory for a codimension $k$ stratum. It is expected that a coupling on each codimension for general geometric situations is possible.
\end{rem}

\subsection{Algebraic and geometric structure for the quantization in higher codimension}
\label{subsec:algebraic_structure_for_the_quantization_in_higher_codimension}

\subsubsection{Deformation Quantization Picture}
Let us denote by $\E_k$ the topological operad of little $k$-dimensional disks and let $\mathbb{P}_k$ denote the operad controlling $(1-k)$-shifted (unbouded) Poisson dg algebras \cite{Lurie2017}.
It is known that deformation quantization of $\mathbb{P}_1$-algebras corresponds to $\E_1$-algebras \cite{Kontsevich1999}, which is the same as an $A_\infty$-algebra (associative algebra).
The higher codimension picture for deformation quantization is related to the higher version of the \emph{Deligne conjecture} \cite{Kontsevich1999}. The Deligne conjecture, which is related to usual deformation quantization, states that there is a natural action of an $\E_2$-algebra over the category of chains complexes to the Hochschild cohomology generated by an arbitrary associative algebra (see e.g. \cite{McClureSmith1999,KontsevichSoibelman2000} for a proof). This can be generalized to the $\E_k$ operad. 
Using the fact that the $\E_k$ operad is formal in the category of chain complexes, i.e. equivalent to its homology and that its homology is given by the $\mathbb{P}_k$ operad, there is an equivalence between the $\E_k$ and the $\mathbb{P}_k$ operad  \cite{Tamarkin1998,Tamarkin2003,Kontsevich1999,FresseWillwacher2020}.
% and for $k$-shifted Poisson structures ($\mathbb{P}_k$-algebras) to $\E_k$-algebras through methods of factorization algebras \cite{CostelloGwilliamVol1,CalaquePantevToenVaquieVezzosi2017}. This is due to the fact that the $\E_k$ operad is formal \cite{Tamarkin2003,Kontsevich1999,FresseWillwacher2020}, i.e. equivalent to its homology, in the category of chain complexes, hence it is equivalent to a $\mathbb{P}_k$-algebra for all $k\geq 2$.
Thus for all $k\geq 2$, there exists a deformation quantization for a $\mathbb{P}_k$-algebra and there is a canonical Lie bracket $[\enspace,\enspace]_{\E_k}$ on an $\E_k$-algebra which corresponds to a $(1-k)$-shifted Poisson structure through the equivalence. One can then view $\calO_{loc}(\calF^{\de^k})[\![\hbar]\!]$ as an $\E_k$-algebra endowed with a dg structure induced by the differential 
\begin{equation}
Q^{\de^k}_\hbar:=\left[\calS^{\de^k}_\hbar,\enspace\right]_{\E_k}.
\end{equation} 
The higher shifted analogue of the quantum master equation is then given by 
\begin{equation}
    \label{eq:higher_QME}
    \left[\calS^{\de^k}_\hbar,\calS^{\de^k}_\hbar\right]_{\E_k}=0.
\end{equation}
% {\textcolor{red}{(Include here discussion on BV and BD algebras)}}
Note that for $k=0$, we obtain a $+1$-shifted Poisson structure $(\enspace,\enspace):=[\enspace,\enspace]_{\E_0}$ as in the construction of a BV algebra, so there is an equivalence of the operad $\mathbb{BV}$, which controls the BV algebra structure, to the homology of the \emph{framed} $\E_0$-operad. Note the difference of having a framing, i.e. we also allow \emph{rotations} of little disks. The rotations indeed correspond to the BV Laplacian $\Delta$ in the homology of the operad over chain complexes.  

In the category of chain complexes we can observe a chain complex of observables $\calO_{loc}(\calF)$ with differential $Q=(\calS,\enspace)$. In fact, we will have the operations given by the differential $Q$ (degree $+1$), usual multiplication of functions (degree $0$), the Poisson structure $(\enspace,\enspace)$ (degree $+1$) and the BV Laplacian $\Delta$ (degree $+1$) such that for all $f,g\in\calO_{loc}(\calF)$ we have
\[
\Delta(fg)=\Delta fg+(-1)^{\gh f}f\Delta g+(-1)^{\gh f}(f,g).
\]
Deforming this operad will lead to the $\mathbb{BD}_0$ operad which is basically the same as the $\mathbb{BV}$ operad without emphasizing the BV Laplacian $\Delta$ but, for the purpose of QFT, include it into the differential and where everything is now over chain complexes of $\mathbf{K}[\![\hbar]\!]$-modules, i.e. we consider the algebra of deformed observables $\calO_{loc}(\calF)[\![\hbar]\!]$ together with the bracket $\hbar(\enspace,\enspace)$ and a differential $D$ of degree $+1$ such that for all $f,g\in\calO_{loc}(\calF)[\![\hbar]\!]$ we have 
\[
D(fg)=D fg+(-1)^{\gh f}fD g+(-1)^{\gh f}\hbar(f,g).
\]
In application to QFT we have $D=Q+\hbar\Delta$ and one can check that we indeed have $D^2=0$ (see also \cite{CostelloGwilliamVol2} for the construction of Beilinson--Drinfeld algebras in connection to QFT and factorization algebras). 

\begin{rem}
It is important to mention that the usual convention for the degrees in the definition of a BV algebra is the one where the Poisson bracket $(\enspace,\enspace)$ is of degree $-1$ and hence $\Delta$ is also of degree $-1$. In this case the $\mathbb{BV}$ operad is equivalent to the homology of the framed $\E_2$ operad (little $2$-disks) \cite{Getzler1994}.
However, for $\mathbb{BD}_0$-algebras one still requires that the bracket is of degree $+1$, so one assigns a weight of $+1$ to $\hbar$. 
\end{rem}

For $k=1$, this corresponds to usual deformation quantization \cite{CostelloGwilliamVol2}. In particular, for a $\mathbb{P}_1$-algebra $(A,\{\enspace,\enspace\})$, the lift to a $\mathbb{BD}_1$-algebra, which is flat over $\mathbf{K}[\![\hbar]\!]$, is the same as a deformation quantization of $A$ in the usual sense. Indeed, to describe $\mathbb{BD}_1$-structures on $A[\![\hbar]\!]$ compatible with the given $\mathbb{P}_1$-structure we can give an associative product on $A[\![\hbar]\!]$, linear over $\mathbf{K}[\![\hbar]\!]$, and which modulo $\hbar$ is given by the commutative product on $A$. The relations in the $\mathbb{BD}_1$ operad imply that the $\mathbb{P}_1$-structure on $A$ is related to the associative product on $A[\![\hbar]\!]$ by
\[
\frac{1}{\hbar}(f\star g- g\star f)=\{f,g\}\,\,\text{mod $\hbar$}.
\]

The higher codimension $k$ version of the quantum master equation follows the picture of $\mathbb{BD}_k$-algebras.

This construction is also consistent with the $k$-dimensional version of the \emph{Swiss-Cheese operad} \cite{Voronov1999,Kontsevich1999} $\mathbb{SC}_{k,k-1}$ which couples the $\E_k$ operad to the $\E_{k-1}$ operad \cite{Kontsevich1999,Thomas2016,Markarian2020} by an action of $\E_k$-algebras on $\E_{k-1}$-algebras. Describe it as an operad of sets. This colored operad has two colors: points may be in the bulk or on the boundary. The set of colors is a poset, that is a category, rather than a set, and there are only operations compatible with this structure. The Swiss-Cheese operad is important when dealing with the coupling in contiguous codimension. A \emph{Swiss-Cheese algebra} (i.e. an algebra over the Swiss-Cheese operad $\mathbb{SC}_{k,k-1}$) consists of a triple $(A,B,\rho)$ such that $A$ is a(n) (framed) $\E_k$-algebra, $B$ a(n) (framed) $\E_{k-1}$-algebra and $\rho\colon A\to \mathrm{HC}^\bullet_{\mathsf{Disk}_{k-1}^\mathrm{fr}}(B)$ the coupling action. 
Here $\mathrm{HC}^\bullet_{\mathsf{Disk}_{k-1}^\mathrm{fr}}(B)$ denotes the \emph{Hochschild cochain object} of the $\E_{k-1}$-algebra $B$.
% Here we have denoted by $\mathsf{Mod}^{\E_{k-1}}_B(B,B)$ the $\infty$-category given by $B$-linear morphisms of framed $\E_{k-1}$-$B$-modules, i.e. morphisms of modules over a(n) (framed) $\E_{k-1}$-algebra $B$. 
In fact, $\mathrm{HC}^\bullet_{\mathsf{Disk}_{k-1}^\mathrm{fr}}(B)$ carries the structure of an $\E_k$-algebra. 
Thus, $\rho$ is a map of (framed) $\E_k$-algebras.
% In particular, the target category is given by the $\E_{k-1}$-linear functor category $\mathsf{Mod}^{\E_k}_B(B,B)=\mathsf{Fun}^\otimes_{\E_{k-1}}(\E_{k-1},\calC^\otimes)$. 

Unfortunately, it can be shown that the Swiss-Cheese operad is not formal \cite{Livernet2015} (still, one can show that there is a higher codimension version of the Swiss-Cheese operad which in fact is indeed formal \cite{Idrissi2018}). However, there is an equivalence of $\mathbb{SC}_{k,k-1}$ to a classical notion for coupling contiguous codimensions by the $\mathbb{P}_{k,k-1}$ operad (see \cite{MelaniSafronov2018_2} for the definition of $\mathbb{P}_{k,k-1}$ and detailed discussions). Let us briefly discuss the coupling in each codimension from the point of view of \emph{factorization homology for stratified spaces} as in \cite{AyalaFrancisTanaka2017}. 
A very good introduction to factorization homology in the topological field theory setting is \cite{Tanaka2020}.
Similarly, as the contiguous codimensions can be coupled together by the Swiss-Cheese operad through the coupling action $\rho$, we can extend this to a generalization for the coupling in each codimension. Denote by $\mathsf{Disk}^\mathrm{fr}_{d}$ the $\infty$-category with objects given by finite disjoint unions of framed $d$-dimensional disks and morphisms being smooth embeddings equipped with a compatibility of framings. A $\mathsf{Disk}_d^\mathrm{fr}$-algebra (or equivalently framed $\E_d$-algebra) $A$ in a monoidal symmetric $\infty$-category $\calC^\otimes$ is a symmetric monoidal functor $A\colon \left(\mathsf{Disk}_d^\mathrm{fr}\right)^{\coprod}\to \calC^\otimes$, where $\coprod$ denotes the symmetric monoidal structure given by disjoint union of topological spaces.
Denote by $\mathsf{Mnfld}_d^\mathrm{fr}$ the $\infty$-category whose objects are smooth framed $d$-dimensional manifolds and whose morphisms consists of all smooth embeddings equipped with a compatibility of framings. 
% In fact, one can construct $\mathsf{Mnfld}_d^\mathrm{fr}$ as a pullback diagram through the tangent bundle construction for the non-framed category $\mathsf{Mnfld}_d$ by considering the forgetful functor $\mathsf{Mnfld}_d^\mathrm{fr}\to \mathsf{Mnfld}_d$.  
% In particular, $\mathsf{Disks}_d^\mathrm{fr}$ is the full subcategory of $\mathsf{Mnfld}_d^\mathrm{fr}$ obtained by a similar pullback as $\mathsf{Mnfld}_d^\mathrm{fr}$.

Let $\calC^\otimes$ be a symmetric monoidal $\infty$-category admitting all sifted colimits. Fix an $\E_d$-algebra in $\calC^\otimes$, which we will denote by $A$. Then \emph{factorization homology with coefficients in $A$} is the left Kan extension, 
\[
\begin{tikzcd}
\mathsf{Disk}^{\mathrm{fr}}_{d} \arrow[d,hook] \arrow[r, "A"]      & \mathcal{C}^\otimes \\
\mathsf{Mnfld}^{\mathrm{fr}}_{d} \arrow[ru, "\int A"', dashed] &            
\end{tikzcd}
\]
and for a framed $d$-manifold $X\in \mathsf{Mnfld}_d^{\mathrm{fr}}$, we denote by 
\[
\int_XA
\]
the \emph{factorization homology of $X$ with coefficients in $A$}. It is called \emph{homology} since it satisfies the \emph{Eilenberg--Steenrod axioms} \cite{EilenbergSteenrod1945} for a homology theory. Note that we can also fix a manifold $Y\in\mathsf{Mnfld}^\mathrm{fr}_{n-d-1}$ and consider a map 
\[
\int_Y\colon \mathsf{Alg}_{\mathsf{Disk}^\mathrm{fr}_n}(\mathcal{C}^\otimes)\to\mathsf{Alg}_{\mathsf{Disk}^\mathrm{fr}_{d+1}}(\calC^\otimes). 
\]
More general, for a symmetric monoidal $\infty$-category $\calC^\otimes$ and an $\infty$-category $\mathsf{B}$ of \emph{basics}, one can define the \emph{absolute factorization homology} to be the left adjoint to 
\[
\begin{tikzcd}
{\displaystyle\int\colon} \mathsf{Alg}_{\mathsf{Disk}(\mathsf{B})}(\calC^\otimes) \arrow[rr, dashed, bend left] &  & {\mathsf{Fun}^\otimes(\mathsf{Mnfld}(\mathsf{B}),\mathcal{C}^\otimes).} \arrow[ll, bend left]
\end{tikzcd}
\]
The basics for our construction are given by framed disks of a given dimension. Moreover, for our purpose we want to consider the category $\calC^\otimes$ to be given by chain complexes. 
For the situation where $\mathsf{B}=\mathsf{Disk}_{d\subset n}^{\mathrm{fr}}$ is defined such that $\mathsf{B}$-manifolds are framed $n$-dimensional manifolds with a framed $d$-dimensional properly embedded submanifold such that the framing splits along this submanifold we have the following:

For an $\E_n$-algebra $A$ and and $\E_d$-algebra $B$, we can observe an action \cite[Proposition 4.8]{AyalaFrancisTanaka2017}
\[
\rho\colon \int_{S^{n-d-1}}A\to \mathrm{HC}^\bullet_{\mathsf{Disk}^{\mathrm{fr}}_d}(B),
\]
where $S^{n-d-1}$ denotes the $(n-d-1)$-sphere.
Here we are considering a $d$-dimensional defect sitting inside an $n$-manifold as locally described by the basics category $\mathsf{Disk}^\mathrm{fr}_{d\subset n}$.
Note that $\rho$ has to be a map of $\E_{d+1}$-algebras.
This can be regarded as a \emph{higher} version of Deligne's conjecture.
In particular, for $d=n-1$, we can obtain the action given by 
\[
\rho\colon \int_{S^0}A=A\otimes A^{\mathrm{op}}\to \mathrm{HC}^\bullet_{\mathsf{Disk}^{\mathrm{fr}}_{n-1}}(B). 
\]
Note that this is not exactly the Swiss-Cheese case since here the domain is given by $A\otimes A^\mathrm{op}$ where as in the Swiss-Cheese case it is given by $A$. Intuitively, this is because the opposite half of the bulk is missing (see Figure \ref{fig:difference}).

\begin{figure}[ht]
\centering
\begin{tikzpicture}[scale=0.6]
\draw[color=gray!20,fill=gray!20] (-2,0) rectangle (-0.1,4);
\draw[color=gray!20,fill=gray!20] (4,0) rectangle (5.9,4);
\draw[color=gray!20,fill=gray!20] (0.1,0) rectangle (1.9,4);
\draw (0,0) -- (0,4);
\draw (6,0) -- (6,4);
\node[label=below: $A$] (A) at (-1,2.5) {};
\node[label=below: $A^\mathrm{op}$] (Aop) at (1,2.5) {};
\node[label=above: $B$] (B) at (0,4) {};
\node[label=below: $A$] (A2) at (5,2.5) {};
\node[label=above: $B$] (B2) at (6,4) {};
\end{tikzpicture}
\caption{Difference between the defect situation and the boundary situation.} 
\label{fig:difference}
\end{figure}
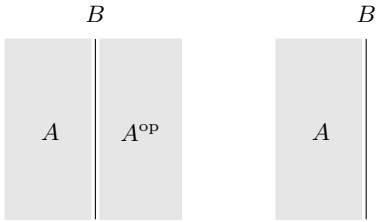

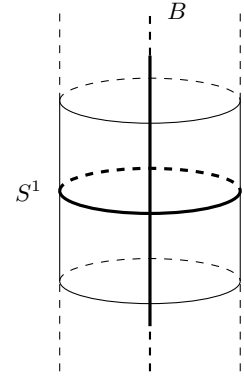
\begin{figure}[ht]
\centering
\begin{tikzpicture}[scale=0.6]
% \draw[color=gray!20,fill=gray!20] (-2,0) rectangle (-0.1,4);
% \draw[color=gray!20,fill=gray!20] (4,0) rectangle (5.9,4);
% \draw[color=gray!20,fill=gray!20] (0.1,0) rectangle (1.9,4);
% \draw (0,0) -- (0,4);
% \draw (6,0) -- (6,4);
% \node[label=below: $A$] (A) at (-1,2.5) {};
% \node[label=below: $A^\mathrm{op}$] (Aop) at (1,2.5) {};
% \node[label=above: $B$] (B) at (0,4) {};
% \node[label=below: $A$] (A2) at (5,2.5) {};
% \node[label=above: $B$] (B2) at (6,4) {};
\draw (-3,0) arc (180:360:2cm and .5cm);
\draw[dashed] (-3,0) arc (180:0:2cm and .5cm);
\draw[very thick] (-3,2) arc (180:360:2cm and .5cm);
\draw[dashed, very thick] (-3,2) arc (180:0:2cm and .5cm);
\draw (-3,4) arc (180:360:2cm and .5cm);
\draw[dashed] (-3,4) arc (180:0:2cm and .5cm);
\draw (-3,0) -- (-3,4);
\draw (1,0) -- (1,4);
\draw[dashed] (-3,4) -- (-3,6);
\draw[dashed] (-3,-2) -- (-3,0);
\draw[dashed] (1,4) -- (1,6);
\draw[dashed] (1,-2) -- (1,0);
\draw[very thick] (-1,-1) -- (-1,5);
\draw[dashed,thick] (-1,-2) -- (-1,-1);
\draw[dashed,thick] (-1,5) -- (-1,6);
\node[label=left: $S^1$] (S1) at (-3,2) {};
\node[label=right: $B$] (B) at (-1,6) {};
\end{tikzpicture}
\caption{Example of a 1-dimensional defect $B$ sitting in $\R^3$. The orthogonal sphere is given by $S^1$ through homotopy equivalence of the tubular neighborhood around it.} 
\label{fig:tubular_neighborhood_situation}
\end{figure}

A nice interpretation for the integration was also given in \cite{Francis2012,CapacciTetik}. The integral $\int_{S^{n-d-1}}A$ is by definition given by $\int_{S^{n-d-1}\times \R^{d+1}}A$, hence an $\E_{d+1}$-algebra since $S^{n-d-1}\times\R^{d+1}$ is an $\E_{d+1}$-algebra by considering the space of embeddings 
\[
\coprod_IS^{n-d-1}\times\R^{d+1}\to S^{n-d-1}\times\R^{d+1}.
\]
This is in fact true for any $(n-d-1)$-manifold not just for $S^{n-d-1}$. Thus, one may think of the action of $A$ on the defect $B$ to be given by a $\int_{S^{n-d-1}}A$-module structure. Consider the local situation when there is a defect $B\subset \R^d$ sitting in $\R^n$. We can think of $S^{n-d-1}$ as the \emph{orthogonal sphere} and $S^{n-d-1}\times \R^{d+1}$ as some tubular neighborhood around $\R^d$ without $\R^d$. Then $\int_{S^{n-d-1}}A$ is the global algebra of observables on this neighborhood. In this interpretation, the action by $\int_{S^{n-d-1}}A$ is equivalent to restricting the bulk algebra to an infinitesimal neighborhood of $B$ where it acts (see Figure \ref{fig:tubular_neighborhood_situation} for an illustration).

For $d=n-2$ we get an action 
\[
\rho\colon\int_{S^1}A=A\otimes_{A\otimes A^{\mathrm{op}}}A\to \mathrm{HC}^\bullet_{\mathsf{Disk}_{n-2}^\mathrm{fr}}(B).
\]
In fact, in \cite{Francis2012,AyalaFrancis2015} it has been proven that for any symmetric monoidal category $\calC^\otimes$ there is an equivalence
\[
\mathrm{Mod}_A^{\mathsf{Disk}^\mathrm{fr}_{n-d}}(\calC^\otimes)\cong \mathrm{Mod}_{\int_{S^{n-d-1}}A}(\calC^\otimes), 
\]
which is used in the proof of \cite[Proposition 4.8]{AyalaFrancisTanaka2017}.
Note that in this picture, the $d$ in $\E_d$ denotes the dimension of the corresponding submanifold (defect). For a classical theory on an $n$-manifold, the codimenion $k$ submanifolds of dimension $d=n-k$, give rise to $\E_d$-algebras, and are endowed with an additional $\mathbb{P}_{k}$-structure. In this case (i.e. with additional structure) all of the previous constructions hold but there will be extra constrains due to the $\mathbb{P}_k$-structure. In particular, it is not clear how the action is expressed in presence of additional structure.

\begin{rem}
An important remark at this point is to emphasize that above mentioned construction for the action of a different codimension is in fact just pairwise. In order to describe the coupling action in each codimension (i.e. not just pairwise) some modification has to be made.
This means that a geometric situation as e.g. in Figure \ref{fig:corner} is not directly covered by the constructions of \cite{AyalaFrancisTanaka2017}. 
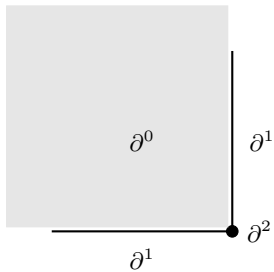
\begin{figure}[ht]
\centering
\begin{tikzpicture}[scale=0.6]
\tikzset{Bullet/.style={fill=black,draw,color=#1,circle,minimum size=0.5pt,scale=0.5}}
\node[Bullet=black, label=right: $\de^2$] (c) at (4,0) {};
\node[label=below: $\de^1$] (d0) at (2,0) {};
\node[label=right: $\de^1$] (d0) at (4,2) {};
\draw[thick] (0,0) -- (4,0);
\draw[thick] (4,0) -- (4,4);
\draw[color=gray!20,fill=gray!20] (-1,0.1) rectangle (3.9,5);
\node[color=black] (d0) at (2,2) {$\de^0$};
\end{tikzpicture}
\caption{A corner situation.} 
\label{fig:corner}
\end{figure}
Note that the coupling depends on the given geometry and how the defects are relative to each other. There has to be some compatibility of the action of the $\de^0$ part to the $\de^1$ parts where these have to have again a coupling to the $\de^2$ part. This has to be compatible with the action of the $\de^0$ to the $\de^2$ part. However, the extension to a situation as in Figure \ref{fig:corner} is known and a coupling on a general stratified space is expected to be possible \cite{CapacciTetik}.
\begin{figure}[ht]
\centering
\begin{tikzpicture}[scale=0.6]
\tikzset{Bullet/.style={fill=black,draw,color=#1,circle,minimum size=0.5pt,scale=0.5}}
\node[Bullet=black] (1) at (0,1) {};
\node[Bullet=black] (2) at (6,3) {};
\draw[thick] (-2,0) -- (8,0);
\filldraw[color=black!,fill=gray!20] (2,3) circle (1cm);
\draw (6,1) circle (0.5cm);
\end{tikzpicture}
\caption{Situation of defects sitting inside some $n$-manifold for which the mentioned construction applies.} 
\label{fig:defects}
\end{figure}
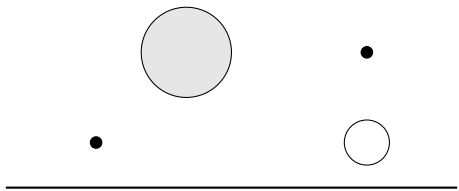
\end{rem}

\subsubsection{Geometric Quantization Picture}
The phase space of an $n$-dimensional classical field theory associated to a closed $d$-dimensional submanifold is endowed with a $(n-d-1)$-shifted symplectic structure \cite{PantevToenVaquieVezzosi2013}. This also makes sense for the case when $d=n$, i.e. for the case of a $(-1)$-shifted symplectic structure which corresponds to the case of the BV formalism. Let us briefly consider the example of 3-dimensional classical Chern--Simons theory as in \ref{subsubsec:Chern-Simons} \cite{Safronov2020}. For a compact Lie group $G$ we associate the phase space which is given by the moduli space of flat $G$-connections. If $\Sigma$ is a closed oriented $3$-manifold, the phase space is given by the critical locus of the Chern--Simons functional, so it carries the induced BV symplectic structure. If $\Sigma$ is a closed oriented 2-manifold, the phase space is endowed with the Atiyah--Bott symplectic structure \cite{AtiyahBott1983}. If $\Sigma=S^1$ (1-dimensional closed compact oriented manifold), the phase space is the stack of conjugacy classes $[G/G]$ and the corresponding 1-shifted symplectic structure is given in terms of the canonical 3-form on $G$. If $\Sigma=\mathrm{pt}$ (0-dimensional case), the phase space is given by the classifying stack $\mathrm{B}G=[\mathrm{pt}/G]$ and the corresponding 2-shifted symplectic structure is given in terms of the invariant symmetric bilinear form on the Lie algebra $\mathfrak{g}:=\mathrm{Lie}(G)$ used in the definition of Chern--Simons theory (Killing form).

The analog of geometric quantization for $k$-shifted symplectic structures uses the notion of \emph{higher categories} \cite{Lurie2017} and \emph{derived algebraic geometry} \cite{Toen2014,PantevToenVaquieVezzosi2013}. It was recently shown that it corresponds to the notion of an $(\infty,k)$-category \cite{Safronov2020}. Let us give a bit more insights on this construction. Recall that the data for geometric quantization of a symplectic manifold $(M,\omega)$, is given by a \emph{prequantization}, i.e. a line bundle (usually called \emph{prequantum line bundle}) $(\mathscr{L},\nabla)$ together with a connection on $M$ with curvature $\omega$, and a \emph{polarization}, i.e. a Lagrangian foliation $\mathscr{F}\subset TM$ which is a subbundle closed with respect to the Lie bracket of vector fields which is Lagrangian with respect to $\omega$ \cite{Kir85,Wood97,Moshayedi2020_geomquant}. The constructed vector space is then given by the space of $\nabla$-flat sections $\Gamma_{\mathrm{flat}}(M,\mathscr{L})\subset\Gamma(M,\mathscr{L})$ of $\mathscr{L}$ along the foliation $\mathscr{F}$. In the $k$-shifted case, one defines the analog of a prequantum line bundle, called a \emph{prequantum $k$-shifted Lagrangian fibration}, to be given by a $k$-gerbe $\mathscr{G}$ on the base manifold together with an extension of the natural relative flat connection on the pullback of $\mathscr{G}$ to the fiber to a connective structure $\nabla$. The polarization is encoded in a $k$-shifted Lagrangian foliation (see \cite{ToenVezzosi2020,BorisovSheshmaniYau2019}). For a 1-shifted symplectic structure we get the $\infty$-category $\mathrm{QCoh}^\mathscr{G}$ of twisted quasi-coherent sheaves on the base and for the 2-shifted case the $\infty$-category of quasi-coherent sheaves on a certain category and so on, such that the output for $k$-shifted structures gives an $(\infty,k)$-category (see \cite{Safronov2020} for a detailed construction).

Note that this indeed is compatible with the quantum BV-BFV construction which assigns a chain complex $(\calH,\Omega)$ to the 0-shifted case on the boundary (rather than a vector space). It is also compatible with the BV construction in the bulk, i.e. for a $(-1)$-shifted symplectic structure. Namely, the constructions as in \cite{S} give rise to a ``geometric quantization''. The expectation value, assigning to an observable a value in the ground field through a path integral of a half-density (usually $\exp(\I\calS/\hbar)\sigma$ where $\sigma$ is some $\Delta$-closed reference half-density) over a chosen Lagrangian submanifold (gauge-fixing), is considered as some analog for the polarization. This can be seen by an extension of the result in \cite{Severa2006}. The original form states that for an odd symplectic supermanifold $(\calF,\omega)$ there is a quasi-isomorphism between the complex $(\Omega^\bullet(\calF),\omega\land)$, consisting of differential forms on $\calF$ endowed with the differential given by wedging with $\omega$, and the complex $(\mathrm{Dens}^{\frac{1}{2}}(\calF),\Delta)$, consisting of half-densities on $\calF$ endowed with the BV Laplacian (This is obviously related to the BV theorem as in \cite{S}). 
Moreover, it states that the de Rham differential vanishes on the cohomology of $(\Omega^\bullet(\calF),\omega\land)$ and that the BV Laplacian is given by $\Delta=(\dd\circ(\omega\land)^{-1}\circ\dd)$.
In fact, it is the third differential in the spectral sequence of the bicomplex $(\Omega^\bullet(\calF),\omega\land,\dd)$ and all higher differentials are zero. The shifted analog states a similar quasi-isomorphism \cite[Proposition 2.34]{Safronov2020}. 
Hence, we want to extend the outcome to ``dg $(\infty,k)$-categories''.
Let us first talk about the 1-categorical picture and give an informal construction for it. 
One can regard a dg $k$-category to be a $k$-category $\mathcal{C}$ \cite{Baez2005} for which each set of morphisms $\Hom(X,Y)$ between two objects $X,Y\in \mathcal{C}$ forms a dg module, i.e. it is given by a direct sum
\[
\Hom(X,Y)=\bigoplus_{n\in \mathbf{Z}}\Hom_n(X,Y),
\]
endowed with a differential
\[
\dd_\mathcal{C}^{(X,Y)}\colon\Hom_n(X,Y)\to \Hom_{n+1}(X,Y). 
\]
Composition of morphisms is given by maps of dg modules
\[
\Hom(X,Y)\otimes \Hom(Y,Z)\to \Hom(X,Z),\quad \forall X,Y,Z\in \calC
\]
satisfying some additional relations \cite{Drinfeld2004,Toen2011}.
One should think of $\Hom$ as the space of 1-morphisms. Denote by $\Hom^{(k)}$ the space of $k$-morphisms, which again forms a dg module. Below an illustration of a 2-morphism $\alpha$ between morphisms $f,g\in \Hom^{(1)}(X,Y)$ and a 3-morphism $\Gamma$ between two 2-morphisms $\alpha,\beta$.
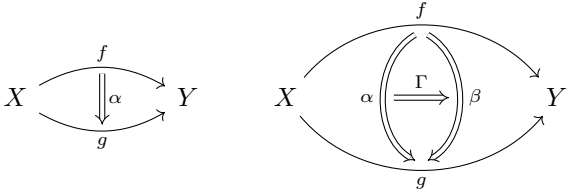
\begin{center}
\begin{tikzcd}
  X \arrow[rr, bend left, "f", ""{name=U,inner sep=1pt,below}]
  \arrow[rr, bend right, "g"{below}, ""{name=D,inner sep=1pt}]
  & & Y
  \arrow[Rightarrow, from=U, to=D, "\alpha"]
\end{tikzcd}\qquad
\begin{tikzcd}[column sep=3cm]
X 
  \arrow[bend left=50]{r}[name=U,below]{}{f} 
  \arrow[bend right=50]{r}[name=D]{}[swap]{g}
& 
Y 
  \arrow[Rightarrow,to path={(U) to[out=-150,in=150] node[auto,swap] {$\scriptstyle\alpha$} coordinate (M) (D)}]{}
  \arrow[Rightarrow,to path={(U) to[out=-30,in=30] node[auto] {$\scriptstyle\beta$} coordinate (N)  (D)}]{}
  \arrow[Rightarrow,to path={([xshift=4pt]M) -- node[auto] {$\scriptstyle\Gamma$} ([xshift=-4pt]N)}]{}
\end{tikzcd}
\end{center}
We require that they satisfy the same conditions as $\Hom=\Hom^{(1)}$, i.e. for two $(k-1)$-morphisms $f,g$, we want that the space of $k$-morphisms between them is given by a direct sum
\[
\Hom^{(k)}(f,g)=\bigoplus_{n\in\mathbf{Z}}\Hom_n^{(k)}(f,g)
\]
endowed with a differential
\[
\dd_\mathcal{C}^{(f,g)}\colon \Hom^{(k)}_n(f,g)\to \Hom^{(k)}_{n+1}(f,g).
\]
The composition of $k$-morphisms should then, similarly as for higher categories, satisfy some Stasheff pentagon identity \cite{Stasheff1963No1,Stasheff1963No2}. 
Formally, one can construct a \emph{strict} dg $k$-category as an iteration (similarly as for defining higher categories) of enrichments over the category of chain complexes. That is, one defines a dg $k$-category as a category enriched over dg $(k-1)$-categories.
This is more or less straightforward since the category of chain complexes can be endowed with a symmetric monoidal structure. The more interesting notion in this setting is the \emph{non-strict} version. There one has to start with an $\infty$-category for the enrichment instead of just chain complexes. Although there should not be any obstacles in the construction, this will be rather involved. The diagram below illustrates the quantization for higher codimensions where we denote by $\mathsf{Ch}$ the category of chain complexes, by 
$\mathsf{Alg}_{\mathbb{P}_k}^{\E_d}(\mathsf{Ch}):=\mathsf{Alg}_{\mathbb{P}_k}(\mathsf{Alg}_{\E_d}(\mathsf{Ch}))$
the category of $\mathbb{P}_k$-algebras over $\E_d$-algebras in $\mathsf{Ch}$ and by $\mathsf{dgCat}_{(\infty,k)}$ the category of dg $(\infty,k)$-categories.
One should think of the horizontal arrows as passing to higher codimension and not as a functor in particular. The quantum picture on the level of deformation quantization focuses on the algebraic structure (shifted Poisson structure) on the space of (higher codimension) observables, whereas the picture on the level of geometric quantization focuses on the geometric structure induced by the space of (higher codimension) boundary fields, namely its (shifted) symplectic manifold structure.

\begin{widetext}
\begin{center}
\begin{tikzcd}
\hline \\
 \hdots& \arrow[l]\de^k & \arrow[l]\hdots & \arrow[l]\de^2& \arrow[l]\de^1& \arrow[l]\de^0\\
\hline\\
\boxed{\mathbf{DefQuant}}\hspace{0.2cm} \hdots & \mathsf{Alg}_{\mathbb{BD}_k}^{\E_{n-k}}(\mathsf{Ch})\arrow[l]\arrow[d] & \hdots\arrow[l] &\mathsf{Alg}_{\mathbb{BD}_2}^{\E_{n-2}}(\mathsf{Ch})\arrow[l] \arrow[d] & \mathsf{Alg}_{\mathbb{BD}_1}^{\E_{n-1}}(\mathsf{Ch}) \arrow[l] \arrow[d]
& \mathsf{Alg}_{\mathbb{BD}_0}^{\E_n}(\mathsf{Ch})\arrow[l] \arrow[d, ""] \\
\boxed{\mathbf{GeomQuant}}\hspace{0.2cm} \hdots & \mathsf{dgCat}_{(\infty,k-1)} \arrow[l] & \hdots \arrow[l]& \mathsf{dgCat}_{(\infty,1)} \arrow[l]& \mathsf{Ch} \arrow[l, ""]
& \arrow[l]\mathbf{K}
\end{tikzcd}
\end{center}
\end{widetext}

\subsection{Obstruction spaces}
Recall that for codimension 0 theories the quantum obstruction space was given by the first cohomology group with respect to the cohomological vector field in the bulk (Theorem \ref{thm:obstruction_BV}) and for coboundary 1 theories it was given by the second cohomology group with respect to the cohomological vector field on the boundary (Theorem \ref{thm:obstruction_BFV}).
% \[\mathrm{H}^1_Q(\calO_{loc}(\calF))\] and for codimension 1 it was given by \[\mathrm{H}^2_{Q^\de}(\calO_{loc}(\calF^\de)).\] 
A natural question is whether the obstruction space for the quantization of codimension $k$ theories is given by \[\mathrm{H}^{k+1}_{Q^{\de^k}}(\calO_{loc}(\calF^{\de^k})).\] This is not clear at the moment. 
We plan to consider this more carefully in the future.

\begin{acknowledgements}
The author would like to thank Alberto Cattaneo, Pavel Mnev, Nicola Capacci, \"Od\"ul Tetik and Eugene Rabinovich for several discussions on this topic. Moreover, he would like to thank Pavel Safronov for discussions and ideas about the algebraic and geometric constructions for higher codimensions. We hope to describe the algebraic ideas of Section \ref{subsec:algebraic_structure_for_the_quantization_in_higher_codimension}, especially the bulk-boundary coupling, in another paper by a precise mathematical formulation. 
Special thanks goes to the anonymous referee who gave valuable comments and suggestions for the improvement of the paper.
This research was supported by the NCCR SwissMAP, funded by the Swiss National Science Foundation, and by the SNF grant No. 200020\_192080.
\end{acknowledgements}

\newpage
%\onecolumngrid
\bibliographystyle{plainurl}
\bibliography{AKSZPSMbibliography}

\end{document}